\documentclass[11pt]{article}
\usepackage[utf8]{inputenc}
\usepackage[margin=1.25in]{geometry}
\usepackage{amsmath}
\usepackage{amssymb}
\usepackage{amsthm}
\usepackage{bbm}

\usepackage{multicol}
\usepackage{multirow}
\usepackage{booktabs}
\usepackage{graphicx}
\usepackage{float}
\usepackage{bm}
\usepackage{multirow}
\usepackage{xargs}
\usepackage{listings}
\usepackage{xcolor}
\usepackage{tabularx}
\usepackage{algorithm}
\usepackage{algcompatible}
\usepackage{rotating}
\usepackage{comment}

\usepackage{bbm}
\usepackage{comment}
\usepackage{natbib}
\usepackage[colorlinks=true, allcolors=blue]{hyperref}

\newtheorem{theorem}{Theorem}

\newtheorem{prop}[theorem]{Proposition}
\newtheorem{corollary}{Corollary}[theorem]

\theoremstyle{definition}
\newtheorem{assumption}{Assumption}

\newcommand{\rightarrowp}{\xrightarrow{p}}

\newcommand{\footremember}[2]{
    \footnote{#2}
    \newcounter{#1}
    \setcounter{#1}{\value{footnote}}
}
\newcommand{\footrecall}[1]{
    \footnotemark[\value{#1}]
} 

\providecommand{\keywords}[1]
{
  \small	
  \textbf{\textit{Keywords---}} #1
}

\title{A Weighted Prognostic Covariate Adjustment Method for Efficient and Powerful Treatment Effect Inferences in Randomized Controlled Trials
\author{Alyssa M. Vanderbeek\footnote{Corresponding author} \footremember{Unlearn.AI}{Unlearn.AI, Inc., San Francisco, CA, USA}, \and Anna A. Vidovszky\footrecall{Unlearn.AI}, \and Jessica L. Ross\footrecall{Unlearn.AI}, \and Arman Sabbaghi\footrecall{Unlearn.AI}, \and Jonathan R. Walsh\footrecall{Unlearn.AI}, \and Charles K. Fisher\footrecall{Unlearn.AI}, \and for the Critical Path for Alzheimer’s Disease\footnote{Data used in the preparation of this article were obtained from the Critical Path for Alzheimer’s Disease (CPAD) Database. In 2008, Critical Path Institute, in collaboration with the Engelberg Center for Health Care Reform at the Brookings Institution, formed the Coalition Against Major Diseases (rebranded to Critical Path for Alzheimer’s Disease (CPAD) consortium in 2018). The consortium brings together patient groups, biopharmaceutical companies, and scientists from academia, the U.S. Food and Drug Administration (FDA), the European Medicines Agency (EMA), the National Institute of Neurological Disorders and Stroke (NINDS), and the National Institute on Aging (NIA). The Critical Path for Alzheimer’s Disease (CPAD) consortium includes over 200 scientists from member and non-member organizations. The data available in the CPAD database has been volunteered by CPAD member companies and non-member organizations.}, \and the Alzheimer’s Disease Neuroimaging Initiative\footnote{Data used in preparation of this article were obtained from the Alzheimer’s Disease Neuroimaging Initiative (ADNI) database (adni.loni.usc.edu). As such, the investigators within the ADNI contributed to the design and implementation of ADNI and/or provided data but did not participate in analysis or writing of this report. A complete listing of ADNI investigators can be found in this document.}, \and European Prevention of Alzheimer’s Disease (EPAD) Consortium\footnote{Data used in preparation of this article were obtained from the Longitudinal Cohort Study (LCS), delivered by the European Prevention of Alzheimer’s Disease (EPAD) Consortium. As such investigators within the EPAD LCS and EPAD Consortium contributed to the design and implementation of EPAD and/or provided data but did not participate in analysis or writing of this report. A complete list of EPAD Investigators can be found in this document.},  and \and the Alzheimer’s Disease Cooperative Study\footnote{Data used in preparation of this manuscript/publication/article were obtained from the University of California, San Diego Alzheimer’s Disease Cooperative Study. Consequently, the ADCS Core Directors contributed to the design and implementation of the ADCS and/or provided data but did not participate in analysis or writing of this report.}}
\date{\today}}

\begin{document}

\maketitle

\newpage

\begin{abstract}
A crucial task for a randomized controlled trial (RCT) is to specify a statistical method that can yield an efficient estimator and powerful test for the treatment effect. A novel and effective strategy to obtain efficient and powerful treatment effect inferences is to incorporate predictions from generative artificial intelligence (AI) algorithms into covariate adjustment for the regression analysis of a RCT. Training a generative AI algorithm on historical control data enables one to construct a digital twin generator (DTG) for RCT participants, which utilizes a participant's baseline covariates to generate a probability distribution for their potential control outcome. Summaries of the probability distribution from the DTG are highly predictive of the trial outcome, and adjusting for these features via regression can thus improve the quality of treatment effect inferences, while satisfying regulatory guidelines on statistical analyses, for a RCT. However, a critical assumption in this strategy is homoskedasticity, or constant variance of the outcome conditional on the covariates. In the case of heteroskedasticity, existing covariate adjustment methods yield inefficient estimators and underpowered tests. We propose to address heteroskedasticity via a weighted prognostic covariate adjustment methodology (Weighted PROCOVA) that adjusts for both the mean and variance of the regression model using information obtained from the DTG. The expectation and variance of the probability distribution for each participant are used as covariates in modeling the mean and variance components for the regression analysis. We prove that our method yields unbiased treatment effect estimators, and demonstrate via comprehensive simulation studies that it can reduce the variance of the treatment effect estimator, maintain the Type I error rate, and increase the power of the test for the treatment effect from $80\%$ to $85\% \sim 90\%$ when the variances from the DTG can explain $5\%  \sim 10\%$ of the variation in the RCT participants' outcomes. We further demonstrate the consistent gains in efficiency and power that can follow from our methodology via re-analyses of data from three RCTs in Alzheimer's disease. Ultimately, our methodology can power the next generation of RCTs and accelerate drug development, as it incorporates features uniquely derived from DTGs to improve the quality of treatment effect inferences in the case of heteroskedasticity while maintaining alignment with regulatory guidance.
\end{abstract}


\newpage

\section{Introduction}
\label{sec:introduction}

Randomized controlled trials (RCTs) are the gold standard for advancing medical theory and developing new treatments \citep{hariton_2018} because they facilitate unbiased inferences on treatment effects and potential complications associated with treatment \citep{food_and_drug_administration_ich_2021}. A crucial task for a RCT is the specification of a statistical method that can yield an efficient (i.e., low variance) treatment effect estimator and a powerful test for the treatment effect. One attractive approach to address this task is to analyze the data from a RCT via covariate adjustment, i.e., a regression model that includes the covariates associated with the outcome along with the treatment indicator as predictors \citep{food_and_drug_administration_adjusting_2023, schuler2022increasing}. The analysis of covariance \citep[ANCOVA,][]{fisher_1934, kempthorne_1952} is an established approach for covariate adjustment in the case of continuous endpoints. Covariate adjustment has been used to great success, and is advantageous over alternatives such as external control arms (ECAs) with respect to satisfying regulatory guidelines on statistical analyses for RCTs because it does not carry the risk of introducing bias or increasing Type I error rates \citep{food_and_drug_administration_2023}. 

Per regulatory guidance, the number of covariates used in covariate adjustment methods should be limited (\citet[FDA,][]{food_and_drug_administration_adjusting_2023} and \citet[EMA,][]{ema_adjusting_2015}). This requirement can be challenging to fulfill given the large number of covariates collected on participants in modern RCTs. Artificial intelligence (AI) offers potential solutions. Consider an AI algorithm trained on historical control data that learns a predictive probability distribution for a participant's control outcome as a function of their baseline covariates. These predicted outcomes (prognostic score) can serve as effective low-dimensional covariates for adjustment in the model. For example, the expected control outcome of a participant in the RCT calculated based on the DTG, referred to as the prognostic score, can serve as a single covariate for adjustment in the analysis. \citet{schuler2022increasing} established a statistical methodology to leverage prognostic scores obtained from an AI algorithm for prognostic covariate adjustment (abbreviated as ``PROCOVA''). The advantages of this method for RCTs are twofold. First, it effectively summarizes a high-dimensional covariate vector into a scalar that can be highly correlated with the observed outcome (because the prognostic score is constructed from an AI algorithm optimized to predict the control outcome). Second, as the AI algorithm is trained on historical, out-of-sample control data, adjusting for the prognostic score will not bias inferences, increase Type I error rates, or decrease confidence interval coverage rates \citep{schuler2022increasing}. Accordingly, adjusting for the prognostic score via a regression model can reduce the variance of the treatment effect estimator and increase the power of the test for the treatment effect in a regulatory acceptable manner. The EMA qualified PROCOVA as ``an acceptable statistical approach for primary analysis'' of Phase 2/3 RCTs with continuous endpoints \citep{ema_procova_2022}.

A critical assumption in covariate adjustment is homoskedasticity, or constant variance of the outcomes conditional on the covariates \citep[p.~12]{weisberg_2014}. Heteroskedasticity, or nonconstant variance of the outcomes, can arise due to lurking variables \citep{joiner_1981} or the nature of the endpoint \citep[p.~171]{weisberg_2014}. Although covariate adjustment methods that are based on homoskedasticity can yield unbiased and consistent estimators for treatment effects when this assumption is violated, failing to account for nonconstant variance typically yields inefficient inferences and underpowered tests for the treatment effect \citep{romano_resurrecting_2016}. 

One could instead fit a weighted regression model where the participants' observations are weighted inversely to their inferred variances, so that participants with smaller inferred variances are given more weight whereas participants with larger inferred variances are given less weight. \citet{davidian_carroll_1987} and \citet{carroll_et_al_1988} investigated the frequentist properties of weighted regression for inferring participant-level variances as a function of their covariates. \citet{romano_resurrecting_2016} proposed an adaptive approach for weighted linear regression. However, there are challenges in incorporating such methods in RCTs. First, there may be many covariates to select from when modeling participants' variances, and attempting to perform variable selection and incorporate multiple selected covariates into the participant-level variance model for weighted regression could make inferences unstable. Second, including multiple covariates in the variance model could make it difficult for non-statisticians to interpret the resulting analyses of a RCT, especially as they may not be as familiar with the variance as with the mean of a distribution. Finally, the inclusion of multiple covariates into the participant-level variance model contradicts FDA and EMA guidance on limiting the number of covariates in an analysis. These complications with weighted regression help to explain the limited application of variance modeling for addressing violations of the homoskedasticity assumption in RCTs. Instead of variance modeling, researchers use heteroskedastic consistent \citep[HC,][]{white_1980} standard errors for treatment effect inferences, which are inefficient when the regression model provides a good fit to the RCT data and do not address biases in the inferences when the model is misspecified \citep{freedman_2006}.

We propose to address heteroskedasticity in a RCT via a weighted prognosic covariate adjustment methodology, which we refer to as ``Weighted PROCOVA'', that leverages historical data through a generative AI algorithm for covariate adjustment of both the mean and variance model for participants. The fundamental idea is to leverage additional information from the predictive probability distribution of the control outcome as learned by the DTG in the participant-level variance component. More formally, our methodology uses the prognostic score as a covariate in the mean model and the logarithm of the variance of the probability distribution from the DTG as a covariate in the participant-level variance model for weighted linear regression. We refer to the inverse of the variance of the probability distribution from the DTG as a participant's ``personalized precision'', and interpret it as the predicted precision of a participant's outcome conditional on all their covariates. Thus, in our methodology the prognostic score remains as the primary covariate adjustment for the mean, and the logarithm of the personalized precision is the sole covariate for the participant-level variance model. The inverse of the predictions from the fitted participant-level variance model define the weights for our weighted regression analysis of RCT data. 

In contrast to existing weighted regression approaches, our Weighted PROCOVA methodology follows regulatory guidance on covariate adjustment by leveraging an AI algorithm trained on historical, out-of-sample control data to summarize the high-dimensional covariate vector for heteroskedasticity into the scalar personalized precision. The participant-level variance model in a RCT is then specified and fitted in a stable and interpretable manner solely as a function of the personalized precision, preventing the complications associated with adjusting for multiple covariances in the variance model. It yields an interpretable weighting mechanism for a RCT, with the strength of the personalized precisions in explaining heteroskedasticity governing the efficiency and power of treatment effect inferences from the weighted regression model. Additionally, the approach relates the quality of the DTG predictions to the gain in statistical efficiency (i.e., reduced variance). Finally, Weighted PROCOVA enables one to adapt the personalized precisions obtained from a DTG to the RCT dataset so as to prevent incongruities in the weights, such as assigning more weight to participants whose outcomes are not predicted well. 
 
We proceed in Section \ref{sec:background} to review regression, heteroskedasticity, and existing approaches to weighted regression in which participant-level variances (and corresponding weights) are inferred by fitting a regression model to a transformation of squared residuals \citep{davidian_carroll_1987, romano_resurrecting_2016}. Section \ref{sec:methods} describes Weighted PROCOVA for RCTs that leverages AI-generated predictions of control outcomes and participant variances. As we describe in this section, our perspective on weighted regression is fundamentally distinct from that of existing methods in that it does not attempt to adjust for multiple covariates in the participant-level variance model, but instead adjusts solely for a transformation of the personalized precisions. We establish the frequentist properties of the coefficient estimators for both the mean and variance models under our methodology in Section \ref{sec:frequentist_properties}, and provide formulae for the variance reduction associated with our methodology in Section \ref{sec:frequentist_properties_variance_reduction}. The power boost of Weighted PROCOVA is a direct function of its variance reduction. In Section \ref{sec:simulation_studies} we present the results of extensive simulation studies on the performance of our method for several settings, including cases in which the variance model is misspecified in that an important covariate for heteroskedasticity is omitted from the model specification. These simulation studies demonstrate the gains offered by our methodology compared to the case of ignoring heteroskedasticity. Finally, in Section \ref{sec:case_studies} we apply our methodology to real-life data from three RCTs on Alzheimer's disease. These analyses demonstrate that adjusting for both the prognostic score and personalized precision in our methodology does not yield biased treatment effect inferences, and that it can reduce the variance of the treatment effect estimator, and boost the power for testing the treatment effect, compared to the linear regression analysis that adjusts just for the prognostic score and does not account for heteroskedasticity. As we conclude in Section \ref{sec:concluding_remarks}, our weighted regression methodology effectively leverages generative AI to improve the quality of treatment effect inferences without introducing much risk, and thereby can power the next generation of RCTs and accelerate drug development.

\section{Background}
\label{sec:background}

\subsection{Notations and Assumptions}
\label{sec:notation}

We adopt the Neyman-Rubin Causal Model \citep{neyman1923causal, holland1986statistics} to define the super-population average treatment effect of a continuous endpoint as our causal estimand of interest for RCTs. This is a causal estimand as it corresponds to a comparison of the treatment and control potential outcomes for the super-population of participants. This estimand also corresponds to the coefficient for the treatment indicator in a linear regression model \citep[p.~119]{imbens_rubin_2015}. 

For each participant $i = 1, \ldots, N$ in a two-arm RCT, let $w_i \in \{0,1\}$ denote their treatment assignment and $x_i \in \mathbb{R}^L$ their covariate vector. A participant's covariate vector corresponds to their characteristics that are observed either prior to treatment assignment, or after treatment assignment and are known to be unaffected by treatment \citep[p.~15--16]{imbens_rubin_2015}. We consider continuous outcomes $Y$ throughout, and assume that the RCT satisfies the Stable Unit-Treatment Value Assumption \citep[SUTVA,][p.~9--13]{imbens_rubin_2015}. The potential outcome for participant $i$ under treatment $w \in \{0, 1\}$ is denoted $Y_i(w)$.

Following the framework of \citet[p.~116--117]{imbens_rubin_2015}, we consider the participants in a RCT as a random sample from an infinite, super-population of participants. The $\left \{ \left ( Y_1(0), Y_1(1), x_1 \right ), \ldots, \left ( Y_N(0), Y_N(1), x_N \right ) \right \}$ are considered to be independent and identically distributed random vectors, with their probability distribution defined by virtue of random sampling. We use $\left ( Y(0), Y(1), X \right )$ to denote a random vector that follows the same distribution as the $\left (Y_i(0), Y_i(1), x_i \right )$. We denote the expectations of the potential outcomes conditional on the covariates by $\mu_0 \left ( x \right ) = \mathbb{E} \left \{ Y(0) \mid X = x \right \}$ and $\mu_1 \left ( x \right ) = \mathbb{E} \left \{ Y(1) \mid X = x \right \}$, and the marginal (i.e., unconditional) expectations for the potential outcomes by $\mathbb{E} \left \{ Y(0) \right \}$ and $\mathbb{E} \left \{ Y(1) \right \}$. The super-population conditional treatment effect is $\tau \left ( x \right ) = \mathbb{E} \left \{ Y(1) \mid X = x \right \} - \mathbb{E} \left \{ Y(0) \mid X = x \right \}$, and the super-population marginal treatment effect is $\tau = \mathbb{E} \left \{ \tau \left ( X \right ) \right \}$. The latter effect is our causal estimand of interest. 

For each participant, at most one potential outcome can be observed according to their treatment assignment. As such, causal inference is a missing data problem under the Neyman-Rubin Causal Model, and the treatment assignment mechanism in the RCT is fundamental \citep{holland1986statistics}. This mechanism is characterized by the joint probability mass function $p ( w_1, \ldots, w_N \mid$ $Y_1(0), Y_1(1), \ldots, Y_N(0), Y_N(1)$, $x_1, \ldots, x_N )$, which is analogous to a missing data mechanism \citep{imbens_rubin_2015, little_rubin_2019}. By virtue of the design of a RCT the treatment assignment mechanism is unconfounded in that $p \left ( w_1, \ldots, w_N \mid Y_1(0), Y_1(1), \ldots, Y_N(0), Y_N(1), x_1, \ldots, x_N \right )$ does not change as a function of the  $\left ( Y_1(0), Y_1(1), \ldots, Y_N(0), Y_N(1) \right )$, i.e., there are no lurking confounders that are associated with both the treatment assignment and the potential outcomes conditional on the covariates. We consider RCTs in which treatment is assigned completely at random, so that $ \left ( w_1, \ldots, w_N \right )$ is independent of $\left ( x_i, \ldots, x_N \right )$, and in which treatment assignment is probabilistic in that $0 \leq p \left ( w_1, \ldots, w_N \mid Y_1 \left ( 0 \right ), Y_1 \left ( 1 \right ), \ldots, Y_N \left ( 0 \right ), Y_N \left ( 1 \right ), x_1, \ldots, x_N \right ) \leq 1$ for any $\left ( w_1, \ldots, w_N \right ) \in \{0, 1\}^N$ with strict inequality holding for at least one such vector. A probabilistic treatment assignment enables us to consider all participants for the design and analysis of the RCT, and reduces the risk of extrapolation bias when inferring treatment effects. We limit our focus to individualistic treatment assignment mechanisms, in which a participant's treatment assignment does not depend on the covariates or potential outcomes of other participants. This generally holds for classical RCTs, and its violation would complicate the design and analysis of a RCT.

\subsection{Linear Regression}
\label{sec:OLS}

Linear regression is an established methodology for causal inference from RCTs. Under this method, $\mu_0(x_i)$ and $\mu_1(x_i)$ are modeled via a predictor vector $v_i \in \mathbb{R}^{M+1}$ as linear functions of unknown regression coefficients $\beta = \left ( \beta_0, \ldots, \beta_M \right ) \in \mathbb{R}^{M+1}$, and the distributions of the $\left ( Y_i(0), Y_i(1) \right )$ conditional on the $v_i$ are specified as independent Bivariate Normal random variables with mean vector $\left ( \mu_0(x_i), \mu_1(x_i) \right )$ and with the covariance matrix proportional to the $2 \times 2$ identity matrix. The $v_i$ are defined via transformations of the $w_i$ and $x_i$, and the entry in $\beta$ associated with $w_i$ in $v_i$ corresponds to the estimand $\tau$ of interest for inference. Inferences for this estimand are performed via ordinary least squares or maximum likelihood estimation according to the linear model specification.

We assume $N > M+1$ throughout. For $v_i=(1, w_i, x_i)^{\mathsf{T}}$, we have the potential outcome generation mechanism
\begin{equation}
\label{eq:linear_model_potential_outcomes}
Y_i(w) = v_i^{\mathsf{T}}\beta + \epsilon_i(w),
\end{equation}
where the random errors $\epsilon_i(w) \sim \mathrm{N} \left ( 0, \sigma^2 \right )$ are mutually independent conditional on the predictors, and $\sigma^2 > 0$ is the variance of the potential outcomes conditional on the predictors. The model in equation (\ref{eq:linear_model_potential_outcomes}) also follows from a representation of the conditional means using functions of the covariates, treatment indicator, and potential outcomes \citep[p.~122--127]{imbens_rubin_2015}. The observed outcome $y_i$ for participant $i$ is a function of their potential outcomes and treatment indicator according to $y_i = w_iY_i(1) + \left ( 1 - w_i \right )Y_i(0)$. This motivates using equation (\ref{eq:linear_model_potential_outcomes}) to specify the linear regression model for the observed outcomes as
\begin{equation}
\label{eq:linear_model_observed_outcomes}
y_i = v_i^{\mathsf{T}}\beta + \epsilon_i,
\end{equation}
where, again, the $\epsilon_i \sim \mathrm{N} \left ( 0, \sigma^2 \right )$ are mutually independent conditional on the predictors \citep[p.~119, 122--127]{imbens_rubin_2015}. We let $y = \left ( y_1, \ldots, y_N \right )^{\mathsf{T}}$, $\epsilon = \left ( \epsilon_1, \ldots, \epsilon_N \right )^{\mathsf{T}}$, and
\begin{equation}
\label{eq:X_matrix}
\mathbf{V} = \begin{pmatrix} v_1^{\mathsf{T}} \\  v_2^{\mathsf{T}} \\ \vdots \\ v_N^{\mathsf{T}} \end{pmatrix}
\end{equation}
denote the linear regression model matrix for equation (\ref{eq:linear_model_observed_outcomes}), so that $y = \mathbf{V}\beta + \epsilon$. The four key assumptions of this model are linearity with respect to the regression coefficients $\beta$, independent random errors, constant variance $\sigma^2$ of the potential outcomes conditional on the predictors, and Normal random errors. We will always assume that $\mathbf{V}^{\mathsf{T}}\mathbf{V}$ is invertible so as to infer $\beta$. 

The causal estimand $\tau$ corresponds to $\beta_1$ in $\beta$, and can be inferred in a consistent manner via the model in equation (\ref{eq:linear_model_observed_outcomes}) \citep{imbens_rubin_2015}. Inference can be performed via ordinary least squares by finding the values that minimize the loss function
\begin{equation}
\label{eq:ols_loss_function}
L \left ( \beta \right ) = \sum_{i=1}^N \left ( y_i - v_i^{\mathsf{T}}\beta \right )^2 = \left ( y - \mathbf{V}\beta \right )^{\mathsf{T}} \left ( y - \mathbf{V}\beta \right ).
\end{equation}
\noindent Alternatively, $\beta$ and $\sigma^2$ can be inferred via maximum likelihood estimation by setting the loss function as the negative log-likelihood function 
\begin{equation}
\label{eq:ols_negative_log_likelihood}
l \left ( \beta, \sigma^2 \right ) = \frac{N}{2}\mathrm{log} \left ( \sigma^2 \right ) + \frac{1}{2\sigma^2}\sum_{i=1}^N  \left ( y_i - v_i^{\mathsf{T}} \beta \right )^2 = \frac{N}{2}\mathrm{log} \left ( \sigma^2 \right ) + \frac{1}{2\sigma^2} \left ( y - \mathbf{V}\beta \right )^{\mathsf{T}} \left ( y - \mathbf{V} \beta \right ).
\end{equation}
The resulting estimator of $\beta$ under either approach is $\hat{\beta} = \left ( \mathbf{V}^{\mathsf{T}} \mathbf{V} \right )^{-1} \mathbf{V}^{\mathsf{T}}y$, and this estimator is unbiased \citep{weisberg_2014}. The maximum likelihood estimator of $\sigma^2$ is biased in finite samples, and the unbiased estimator is $\widehat{\sigma^2} = \sum_{i=1}^N \left ( y_i - \hat{y}_i^2 \right )^2/(N-M-1)$ where $\hat{y}_i = v_i^{\mathsf{T}}\hat{\beta}$ is the predicted outcome for participant $i$. Both estimators of $\sigma^2$ involve the sum of squared residuals, with the residual for participant $i$ defined as $e_i = y_i - \hat{y}_i$. As proven by \citet[p.~122--127]{imbens_rubin_2015}, $\hat{\beta}_1$ is a \citet{fisher_1922} consistent estimator of $\tau$. 

The super-population perspective that we adopt is distinct from the finite-population perspective in which random sampling of participants is not considered, potential outcomes are fixed (potentially unobserved) numbers, and covariates are not probabilistic. As demonstrated by \citet{freedman2008randomization_2, freedman2008randomization_1}, the finite-population perspective introduces the complication of finite-sample bias for estimation of $\tau$ via $\hat{\beta}_1$ because the incorporation of covariates requires the estimation of additional nuisance parameters (i.e., the regression coefficients for the covariates) \citep[p.~113]{imbens_rubin_2015}. However, this issue does not arise under the super-population perspective.   

The vector of predicted outcomes for the participants is $\hat{y} = \mathbf{V} \left ( \mathbf{V}^{\mathsf{T}} \mathbf{V} \right )^{-1}\mathbf{V}^{\mathsf{T}} y$. The matrix $\mathbf{H} = \mathbf{V} \left ( \mathbf{V}^{\mathsf{T}} \mathbf{V} \right )^{-1}\mathbf{V}^{\mathsf{T}}$, referred to as the ``hat matrix'', will feature in the theory for Weighted PROCOVA in Section \ref{sec:frequentist_properties}. We let $h_{ij}$ denote entry $(i,j)$ of $\mathbf{H}$. The hat matrix is idempotent and symmetric. Each diagonal entry $h_{ii}$ corresponds to the leverage for participant $i$ \citep{weisberg_2014}, with $h_{ii} = \sum_{j=1}^N h_{ij}^2 = h_{ii}^2 + \sum_{j \neq i}^N h_{ij}^2 $, so that $0 \leq h_{ii} \leq 1$. All of the $h_{ij}$ approach zero when all the $h_{ii}$ approach zero, or when all the $h_{ii}$ approach $1$ \citep{meloun_2011}.

\subsection{Weighted Linear Regression}
\label{sec:WLS}

Weighted linear regression addresses heteroskedasticity in the outcomes. The potential outcome generation mechanism in this case is a modification of that in equation (\ref{eq:linear_model_potential_outcomes}) in which the $\epsilon_i(w) \sim \mathrm{N} \left ( 0, \sigma_i^2 \right )$ with each participant $i$ having their individual variance parameter $\sigma_i^2 > 0$. The observed outcomes are similarly modeled as in equation (\ref{eq:linear_model_observed_outcomes}), with the modification that the $\epsilon_i \sim \mathrm{N} \left ( 0, \sigma_i^2 \right )$. 

If the $\sigma_i^2$ were known, then the Best Linear Unbiased Estimator of $\beta$ is obtained via weighted least squares by finding the values that minimize the loss function 
\begin{equation}
\label{eq:wls_loss_function}
L \left ( \beta \right ) = \sum_{i=1}^N \left ( y_i - v_i^{\mathsf{T}} \beta \right )^2/\sigma_i^2 = \left ( y - \mathbf{V}\beta \right )^{\mathsf{T}} \Omega^{-1} \left ( y - \mathbf{V} \beta \right ),
\end{equation}
where $\Omega$ is the $N \times N$ diagonal matrix whose $i$th diagonal entry is $\Omega_{ii} = \sigma_i^2$ \citep[p.~3]{romano_resurrecting_2016}. This loss function is the negative log-likelihood function when the $\sigma_i^2$ are known. Consideration of inference for $\beta$ in the case of known $\sigma_i^2$ motivates a ``plug-in'' inferential strategy in the case of unknown $\sigma_i^2$ that can be distinct from maximum likelihood inference. Specifically, one first estimates $\sigma_i^2$ based on the RCT data, then plugs those estimates $\widehat{\sigma_i^2}$ into $\Omega$ in equation (\ref{eq:wls_loss_function}), and finally derives the estimator $\hat{\beta}$ that minimizes the loss function under fixed $\widehat{\sigma_i^2}$ \citep{davidian_carroll_1987, romano_resurrecting_2016}. The resulting estimator is
\begin{equation}
\label{eq:wls_estimator}
\hat{\beta} = \left ( \mathbf{V}^{\mathsf{T}} \widehat{\Omega}^{-1} \mathbf{V} \right )^{-1} \mathbf{V}^{\mathsf{T}} \widehat{\Omega}^{-1} y,
\end{equation}
\noindent where $\widehat{\Omega}$ is the $N \times N$ diagonal matrix whose $i$th diagonal entry is $\widehat{\Omega}_{ii} = \widehat{\sigma_i^2}$. 

Under the variance modeling approach for weighted linear regression, a transformation of $\sigma_i^2$ is modeled as a function of the covariates to obtain the estimators $\widehat{\sigma_i^2}$ that will be input into equation (\ref{eq:wls_loss_function}). This model is also referred to as the skedastic function model \citep[p.~2]{romano_resurrecting_2016}. We let $u_i \in \mathbb{R}^{L+1}$ denote the predictor vector for the skedastic function model, with the first entry in each $u_i$ equal to $1$. A natural transformation of the $\sigma_i^2$ for the skedastic function model is the logarithmic transformation, with the model being $\mathrm{log} \left ( \sigma_i^2 \right ) = u_i^{\mathsf{T}}\gamma$ for an unknown coefficient vector $\gamma \in \mathbb{R}^{L+1}$. Two other common transformations are the square root transformation $\sqrt{\sigma_i^2} = u_i^{\mathsf{T}}\gamma$ and the identity transformation $\sigma_i^2 = u_i^{\mathsf{T}}\gamma$. The latter two transformations complicate inferences because the $u_i^{\mathsf{T}}\gamma$ are restricted to be positive. A disadvantage of the logarithmic transformation is that the resulting inferences are sensitive to outliers \citep[p.~1083]{davidian_carroll_1987}. As the logarithmic transformation is considered more interpretable in practice \citep[p.~4]{romano_resurrecting_2016}, we utilize it in all the skedastic function models that we consider for weighted linear regression. Under the logarithmic transformation and an estimator $\hat{\gamma}$ for $\gamma$, the estimator for $\sigma_i^2$ is $\widehat{\sigma_i^2} = \mathrm{exp} \left ( u_i^{\mathsf{T}} \hat{\gamma} \right )$.

Maximum likelihood-based inferences for the treatment effect in the case of heteroskedasticity, as through the optimization of the negative log-likelihood function
\begin{equation}
\label{eq:wls_negative_log_likelihood_transformation}
l \left ( \beta, \gamma \right ) = \sum_{i=1}^N u_i^{\mathsf{T}} \gamma + \sum_{i=1}^N \mathrm{exp} \left ( -u_i^{\mathsf{T}}\gamma  \right ) \left ( y_i - v_i^{\mathsf{T}} \beta \right )^2,
\end{equation}
can be advantageous due to the existing theory for maximum likelihood estimation and the asymptotic efficiency of the parameter estimators when the model is well-specified. However, it has a significant disadvantage in that it is sensitive to model misspecification \citep[p.~1083]{davidian_carroll_1987}. Furthermore, optimization of equation (\ref{eq:wls_negative_log_likelihood_transformation}) is neither computationally straightforward nor insightful in terms of enabling one to understand how the quality of the treatment effect inferences is related to the predictors $u_i$. \citet[p.~1082--1084]{davidian_carroll_1987} describe an approach that is less sensitive to model misspecification, and more computationally efficient and insightful than maximum likelihood estimation. This approach involves iterative modeling of a transformation of the squared residuals $e_i^2 = \left ( y_i - \hat{y}_i \right )^2$ from linear regression. Specifically, under this approach and adoption of the logarithmic transformation, the $\mathrm{log} \left ( e_i^2 \right)$ are modeled via a separate regression
\begin{equation}
\label{eq:skedastic_function_model_residuals}
\mathrm{log} \left ( e_i^2 \right ) = u_i^{\mathsf{T}}\gamma + \xi_i,    
\end{equation}
with $\xi_i \sim \mathrm{N} \left ( 0, \psi^2 \right )$ being mutually independent. Alternatively, the loss function is
\begin{equation}
\label{eq:weighted_procova_loss_function}
L \left ( \beta, \gamma \right ) =  \sum_{i=1}^N \left [ \mathrm{log} \left \{ \left ( y_i - v_i^{\mathsf{T}} \beta \right )^2 \right \} - u_i^{\mathsf{T}} \gamma \right ]^2 + \sum_{i=1}^N \mathrm{exp} \left ( -u_i^{\mathsf{T}} \gamma \right ) \left ( y_i - v_i^{\mathsf{T}} \beta \right )^2,
\end{equation}
\noindent and instead of calculating the $(\beta, \gamma)$ values that simultaneously minimize this loss function one divides equation (\ref{eq:weighted_procova_loss_function}) into two components for iterative estimation of one parameter vector conditional on the other. More formally, we let $f_i \left ( \beta, \gamma \right ) = \mathrm{exp} \left ( - u_i^{\mathsf{T}} \gamma \right ) \left ( y_i - v_i^{\mathsf{T}} \beta \right )^2$, so that equation (\ref{eq:weighted_procova_loss_function}) is
\begin{equation}
\label{eq:weighted_procova_loss_function_reformulation}
L \left ( \beta, \gamma \right ) = \sum_{i=1}^N \left [ \mathrm{log} \left \{ f_i \left ( \beta, \gamma \right ) \right \} \right ]^2 + \sum_{i=1}^N f_i \left ( \beta, \gamma \right ).
\end{equation}
\noindent By means of equation (\ref{eq:weighted_procova_loss_function_reformulation}), for any estimates $\gamma^*$ of $\gamma$ and $\beta^*$ of $\beta$ we have that minimization of $\sum_{i=1}^N f_i \left ( \beta, \gamma^* \right )$ is a linear regression problem for $\beta$, and that minimization of $\sum_{i=1}^N \left [ \mathrm{log} \left \{ f_i \left ( \beta^*, \gamma \right ) \right \} \right ]^2$ is a linear regression problem for $\gamma$. Thus, a computationally efficient approach to estimate $(\beta, \gamma)$ is to iterate between the two regression problems given by these two components of equation (\ref{eq:weighted_procova_loss_function_reformulation}). The predictions $\mathrm{exp} \left ( u_i^{\mathsf{T}} \hat{\gamma} \right )$ obtained from the estimator $\hat{\gamma}$ of $\gamma$ under this approach then define the weights that will be input in equation (\ref{eq:wls_loss_function}) according to $\widehat{\Omega}_{ii} = \mathrm{exp} \left ( u_i^{\mathsf{T}} \hat{\gamma} \right )$. 

In practice, one initializes the iterative process by setting $\gamma^{(0)} = ( 0, \ldots, 0 )^{\mathsf{T}}$. In the first iteration under this initialization the estimator for $\beta$ is $\hat{\beta}^{(1)} = ( \mathbf{V}^{\mathsf{T}} \mathbf{V} )^{-1} \mathbf{V}^{\mathsf{T}} y$, and the updated estimator for $\gamma$ is $\hat{\gamma}^{(1)} = ( \mathbf{U}^{\mathsf{T}} \mathbf{U} )^{-1} \mathbf{U}^{\mathsf{T}} \mathrm{log} \{ ( y - \mathbf{H}y )^2 \}$, where $\mathbf{U}$ is the $N \times (L+1)$ matrix whose $i$th row is $u_i^{\mathsf{T}}$ and the logarithmic operation in this expression is performed entry-wise. These two steps can continue in multiple iterations if desired, with the mapping function and solution path for the procedure ultimately being a function of $\gamma$. Specifically, if we let $\mathcal{M}: \mathbb{R}^{M+1} \times \mathbb{R}^{L+1} \rightarrow \mathbb{R}^{M+1} \times \mathbb{R}^{L+1}$ denote the mapping function $\mathcal{M} ( \hat{\beta}^{(t-1)}, \hat{\gamma}^{(t-1)} ) = ( \hat{\beta}^{(t)}, \hat{\gamma}^{(t)} )$ from iteration $t-1$ to $t$, then $\hat{\beta}^{(t)} = \{ \mathbf{V}^{\mathsf{T}} ( \widehat{\Omega}^{(t-1)})^{-1} \mathbf{V} \}^{-1} \mathbf{V}^{\mathsf{T}} ( \widehat{\Omega}^{(t-1)} )^{-1} y$ and $\hat{\gamma}^{(t)} = \left ( \mathbf{U}^{\mathsf{T}} \mathbf{U} \right )^{-1} \mathbf{U}^{\mathsf{T}}e^{(t-1)}$ where $\widehat{\Omega}^{(t-1)}$ is a $N \times N$ diagonal matrix whose $i$th diagonal entry is $\mathrm{exp} \left ( v_i^{\mathsf{T}}\hat{\gamma}^{(t-1)} \right )$ and $e^{(t-1)}$ is a $N \times 1$ vector whose $i$th entry is $e_i^{(t-1)} = y_i - v_i^{\mathsf{T}} \{ \mathbf{V}^{\mathsf{T}} ( \widehat{\Omega}^{(t-1)} )^{-1} \mathbf{V} \}^{-1} \mathbf{V}^{\mathsf{T}} ( \widehat{\Omega}^{(t-1)} )^{-1} y$. If $\hat{\gamma}^{(t-1)}$ is the maximum likelihood estimate of $\gamma$ from equation (\ref{eq:wls_negative_log_likelihood_transformation}), then $\hat{\beta}^{(t)}$ will be the maximum likelihood estimate of $\beta$.

The computation involved with the approach of \citet{davidian_carroll_1987} is ordinary least squares, which is simpler than the optimization of equation (\ref{eq:wls_negative_log_likelihood_transformation}). Furthermore, as we demonstrate via simulation in Section \ref{sec:simulation_studies_results_boost}, the quality of the regression model in equation (\ref{eq:skedastic_function_model_residuals}) as measured by its coefficient of determination yields insights on the ability of the skedastic function model to reduce the variance of the treatment effect estimator. Similarly, the consideration of squared residuals in this approach is insightful as it indicates that participants with smaller squared residuals should be given more weight than participants with larger squared residuals, and that if the reverse were to occur then the weighted linear regression would result in lower-quality inferences for the treatment effect compared to linear regression. This procedure can be iterative because new residuals can be calculated from a particular weighted least squares solution to fit a model for the logarithmic transformation of the new squared residuals, which yields new weights as inputs for equation (\ref{eq:wls_loss_function}). This iteration is helpful when the estimates for $\beta$ obtained from weighted least squares differ from those obtained via ordinary least squares \citep{davidian_carroll_1987}.

\subsection{The Distributions of Residuals Under Heteroskedasticity}
\label{sec:residuals}

The skedastic function modeling approach that we adopt in Weighted PROCOVA involves modeling of transformations of residuals from linear regression. We summarize the distributional properties for the residuals, squared residuals, and logarithmic transformed squared residuals. These properties will be applied in Section \ref{sec:frequentist_properties} to establish the finite-sample and asymptotic properties of $\hat{\gamma}$ and $\hat{\beta}$. Derivations of these properties are in Appendix \ref{sec:appendix_residuals_proofs}.

The residuals of a linear regression model in the case of heteroskedasticity follow a Multivariate Normal distribution with $\left [ e_i \mid \mathbf{V} \right ] \sim \mathrm{N} \left ( 0, \displaystyle \left ( 1 - h_{ii} \right )^2 \sigma_i^2 + \sum_{k \neq i} h_{ik}^2 \sigma_k^2 \right )$ and $\mathrm{Cov} \left ( e_i, e_j \mid \mathbf{V} \right ) = \displaystyle \sum_{k \neq i, j} h_{ik}h_{jk}\sigma_k^2 - h_{ij} \left ( 1 - h_{ii} \right )\sigma_i^2 - h_{ij} \left ( 1 - h_{jj} \right )\sigma_j^2$. The squared residuals are correlated scaled $\chi_1^2$ random variables. Specifically, letting $\left [ Z_i \mid \mathbf{V} \right ] \sim \mathrm{N} \left ( 0, 1 \right )$ be Normal random variables with
\begin{align*}
\mathrm{Cov} \left ( Z_i, Z_j \mid \mathbf{V} \right ) &= \left ( \left ( 1 - h_{ii} \right )^2 \sigma_i^2 + \sum_{k \neq i} h_{ik}^2 \sigma_k^2 \right )^{-1/2} \left ( \left ( 1 - h_{jj} \right )^2 \sigma_j^2 + \sum_{k \neq j} h_{jk}^2 \sigma_k^2 \right )^{-1/2} \\ 
& \ \ \ \times \left \{ \displaystyle \sum_{k \neq i, j} h_{ik}h_{jk}\sigma_k^2 - h_{ij} \left ( 1 - h_{ii} \right )\sigma_i^2 - h_{ij} \left ( 1 - h_{jj} \right ) \sigma_j^2 \right \},
\end{align*}
\noindent 
then $ \left [ e_i^2 \mid \mathbf{V} \right ] \sim \left [ \left ( \left ( 1 - h_{ii} \right )^2 \sigma_i^2 + \displaystyle \sum_{k \neq i} h_{ik}^2 \sigma_k^2 \right ) Z_i^2 \mid \mathbf{V} \right ]$. Furthermore, the $Z_i^2$ and $e_i^2$ are respectively correlated according to
\begin{align*}
\mathrm{Cov} \left ( Z_i^2, Z_j^2 \mid \mathbf{V} \right ) &= \left ( \left ( 1 - h_{ii} \right )^2 \sigma_i^2 + \sum_{k \neq i} h_{ik}^2 \sigma_k^2 \right )^{-1} \left ( \left ( 1 - h_{jj} \right )^2 \sigma_j^2 + \sum_{k \neq j} h_{jk}^2 \sigma_k^2 \right )^{-1} \\
& \ \ \ \times \left [ 2 \left \{ \displaystyle \sum_{k \neq i, j} h_{ik}h_{jk}\sigma_k^2 - h_{ij} \left ( 1 - h_{ii} \right )\sigma_i^2 - h_{ij} \left ( 1 - h_{jj} \right ) \sigma_j^2 \right \}^2 \right ],
\end{align*}
and $\mathrm{Cov} \left ( e_i^2, e_j^2 \mid \mathbf{V} \right ) = 2 \left \{ \displaystyle \sum_{k \neq i, j} h_{ik}h_{jk}\sigma_k^2 - h_{ij} \left ( 1 - h_{ii} \right )\sigma_i^2 - h_{ij} \left ( 1 - h_{jj} \right ) \sigma_j^2 \right \}^2$.

The logarithmic transformed squared residuals are correlated and distributed according to $\left [ \mathrm{log} \left ( e_i^2 \right ) \mid \mathbf{V} \right ] \sim \mathrm{log} \left ( \left ( 1 - h_{ii} \right )^2 \sigma_i^2 + \displaystyle \sum_{k \neq i} h_{ik}^2 \sigma_k^2 \right ) + \left [ \mathrm{log} \left ( Z_i^2 \right ) \mid \mathbf{V} \right ]$. Their expectations and variances are $\mathbb{E} \left \{ \mathrm{log} \left ( e_i^2 \right ) \mid \mathbf{V} \right \} = \mathrm{log} \left ( \left ( 1 - h_{ii} \right )^2 \sigma_i^2 + \displaystyle \sum_{k \neq i} h_{ik}^2 \sigma_k^2 \right ) - (\gamma_{\mathrm{EM}} + \mathrm{log} \ 2 )$, where $\gamma_{\mathrm{EM}} \approx 0.577$ denotes the Euler-Mascheroni constant (and does not correspond to skedastic function model coefficients), and $\mathrm{Var} \left \{ \mathrm{log} \left ( e_i^2 \right ) \mid \mathbf{V} \right \} = \pi^2/2$.

\subsection{Prognostic Covariate Adjustment}
\label{sec:ML}

Covariate adjustments for RCTs can be performed using transformations of covariates. A transformation that can yield significant variance reduction for the treatment effect estimator is the expected control outcome conditional on the covariates \citep[p.~5--6]{schuler2022increasing}. The application of an optimized AI algorithm, trained on historical control data, to define this transformation can enable one to obtain valid, efficient, and powerful treatment effect inferences via a single adjustment. 

The PROCOVA methodology of \citet{schuler2022increasing} is a linear regression analysis of a RCT in which a single covariate for adjustment, referred to as the prognostic score, is derived via a trained AI algorithm. This methodology consists of five essential steps. First is the creation of the digital twin generator (DTG) by training the AI algorithm on historical control data. The sole inputs for the trained DTG are participant covariates and baseline measures  $x_i$. For each participant $i$ and timepoint after treatment assignment, the output of the DTG is a probability distribution for their control potential outcome. We consider a fixed timepoint post-treatment assignment throughout, and represent the probability distribution by the cumulative distribution function (CDF) $F_{i,0}: \mathbb{R} \rightarrow (0,1)$. In the second step, each CDF $F_{i,0}$ from the DTG is used to calculate the prognostic score $m_i = \int_{-\infty}^{\infty} r dF_{i,0}(r)$ for each participant in the RCT. In practice, prognostic scores are calculated via Monte Carlo, i.e., by obtaining independent samples from each distribution $F_{i,0}$ and calculating the averages of the respective samples. The third step is a linear regression analysis of the RCT in which the vector of predictors $v_i$ for participant $i$ contains $w_i$ and $m_i$. The treatment effect estimator is obtained in the fourth step via the fitted regression model and corresponds to $\hat{\beta}_1$. Finally, in the fifth step, the standard error of the treatment effect is calculated.

It is important to recognize that the first step does not involve any RCT data, is completely pre-specified for the RCT, and that the trained DTG has as its sole inputs the participants' covariates and baseline measures. As such, it does not introduce any bias in treatment effect inferences for the RCT. The DTG is applied to both treated and control participants in a RCT in the second step, effectively summarizing their large number of covariates into a scalar that is correlated with the observed outcomes. In the third step, one could consider an interaction between the treatment indicator and prognostic score, and include other covariates of interest (as well as their interactions with the treatment indicator and prognostic score) as predictor variables. \citet{schuler2022increasing} calculated HC standard errors in the fifth step, specifically, via the HC1 method \citep[p.~6--7]{romano_resurrecting_2016}, based on regulatory guidance regarding covariate adjustment for RCTs \citep{ema_adjusting_2015, food_and_drug_administration_adjusting_2023}. The EMA qualification of this procedure also notes the use of HC standard errors \citep{ema_procova_2022}.

Simulation studies were conducted by \citet[p.~8--10]{schuler2022increasing} to compare and contrast the mean squared error of the treatment effect estimator under PROCOVA and other methods. Across six different settings that were defined in terms of the distributions of covariates and potential outcomes, they observed that covariate adjustment with prognostic scores consistently yields smaller mean squared errors compared to the other methods. These observations demonstrate their theoretical results that PROCOVA attains the minimum possible asymptotic variance among a class of estimators, the uncertainty in the treatment effect estimator is minimized when the prognostic model predicts the participants' control potential outcomes, and that one can realize gains in efficiency even with imperfect prognostic models or in the presence of heterogeneous effects. They also established that PROCOVA decreases the variance of the treatment effect estimator proportional to the squared correlation of the prognostic score with the outcome while guaranteeing unbiasedness, control of Type I error rates, and desirable confidence interval coverage. This enabled \citet[p.~7--8]{schuler2022increasing} to derive a sample size calculation for the design of smaller trials that maintain their original power level via PROCOVA. The variance and sample size reductions achievable via PROCOVA were illustrated via a reanalysis of a previously reported clinical trial studying the effect of docosahexaenoic acid (DHA) on cognition in patients with Alzheimer’s disease \citep{quinn_et_al_2010}, which we shall describe in more detail in Section \ref{sec:case_studies}. Their reanalysis indicated that PROCOVA yields the same inferences as previous statistical models, with comparable (if not smaller) standard errors and narrower confidence intervals. \citet[p.~12]{schuler2022increasing} also performed a prospective power analysis based on the case study (without using any data from the trial), and demonstrated that PROCOVA could lead to a $20\%$ reduction of the control arm sample size while ensuring $80\%$ power for the test of the treatment effect.

PROCOVA unites AI and historical control data to decrease uncertainty in treatment effect inferences for RCTs. The DTG can be implemented by any mathematical or computational means, and \emph{a priori} variance reduction and power boost are cast in terms of goodness-of-fit metrics for DTGs on validation data. The benefits of adjusting via the prognostic score increase as a function of the correlation between the prognostic score and the outcomes. In addition, adjusting based solely on the prognostic score instead of the entire predictor vector $v_i$ uses fewer degrees of freedom, which in turn can increase power when $M$ is large. PROCOVA can be semiparametric efficient, so that the power of a trial using PROCOVA could be greater than or equal to the power of any other trial design that controls the Type I error rate \citep{schuler2022increasing}.

\section{Weighted Prognostic Covariate Adjustment}
\label{sec:methods}

\subsection{Adjustments for Prognostic Scores and Personalized Precisions}
\label{sec:methods_objective}

Weighted PROCOVA utilizes an additional feature that is uniquely derived from the DTG, beyond the prognostic score, to address heteroskedasticity and provide a more efficient and powerful analysis in a RCT. For each participant $i = 1, \ldots, N$ at a fixed timepoint post-treatment, we calculate $s_i^2 = \int_{-\infty}^{\infty} \left ( r - m_i \right )^2 dF_{i,0}(r)$ and utilize $\mathrm{log} \left ( s_i^2 \right )$ as the sole predictor for the skedastic function model. We refer to $ \left ( s_i^2 \right )^{-1}$ as the personalized precision for participant $i$, and recognize that the predictor for the skedastic function model is essentially the logarithm of the personalized precision. Similar to prognostic scores, we calculate the $s_i^2$ via Monte Carlo by obtaining independent samples from each distribution $F_{i,0}$ and calculating the variances of the respective samples. Given the point estimates $\hat{\gamma} = \left ( \mathbf{U}^{\mathsf{T}}\mathbf{U} \right )^{-1} \mathbf{U}^{\mathsf{T}} \mathrm{log} \left \{ \left ( y - \mathbf{Y} y \right )^2 \right \}$ as defined in Section \ref{sec:WLS} for $u_i = \left ( 1, \mathrm{log} \left ( s_i^2 \right ) \right )^{\mathsf{T}}$, we then estimate $\widehat{\sigma}_i^2 = e^{\hat{\gamma}_0} \left ( s_i^2 \right )^{\hat{\gamma}_1}$ and construct $\widehat{\Omega}$ accordingly to infer the treatment effect in terms of $\beta_1$. In order to satisfy regulatory guidelines, inferences on $\beta_1$ are based on the standard error of $\hat{\beta}_1$ as calculated using the combination of the HC1 method with $\widehat{\Omega}$. The prognostic score $m_i$ remains as a covariate in $v_i$. Hence, Weighted PROCOVA uses two features from the DTG, with $m_i$ used in the model for $\mathbb{E} \left ( y_i \mid \mathbf{U}, \mathbf{V} \right )$ and $s_i^2$ used in the model for $\mathrm{Var} \left ( y_i \mid \mathbf{U}, \mathbf{V} \right )$. The estimated variance $\widehat{\sigma}_i^2$ for each participant is proportional to a power of their $s_i^2$ as derived from the DTG. If the DTG is well-calibrated we can expect $\hat{\gamma}_0 \approx 0$ and $\hat{\gamma}_1 \approx 1$, so that the $(i,i)$ entry of $\widehat{\Omega}$ is effectively $s_i^2$. 

Diagnostics can be performed to help one recognize potentially undesirable statistical regimes for the application of Weighted PROCOVA. One set of diagnostics involves plots or statistical tests on the residuals from PROCOVA to evaluate the evidence for heteroskedasticity. Weighted PROCOVA would still be valid under homoskedasticity, but it would be preferable to directly set $\gamma_1 = 0$. Another diagnostic is an examination of the variance of the $s_i^2$. Specifically, when the variance of the $s_i^2$ is small then this feature may fail to provide sufficient information for Weighted PROCOVA to improve the quality of treatment effect inferences. A third diagnostic is the examination of the joint distributions of the inferred weights, personalized precisions, and residuals to confirm that large weights are attached to participants for whom the model is more confident with respect to predicting their outcomes. Weighted PROCOVA is expected to improve the quality of treatment effect inferences when participants with small residuals have large inferred weights. Alternatively, the quality of inferences is expected to decrease when participants with large residuals have large inferred weights. Finally, it is important to diagnose whether $\gamma_1$ is inferred to be negative. In this case, participants with small personalized precision could be given more weight even if their squared residuals are large. This would be undesirable as one should give more weight to participants with large personalized precisions and small squared residuals. 

A key advantage of Weighted PROCOVA is in its incorporation of generative AI. Generative AI can effectively summarize the relationships between the variation in the outcomes and all the covariates to construct a scalar feature $s_i^2$ that is optimized to explain the variation. This advantageous feature can free one from having to analyze all the entries in a high-dimensional covariate vector or consider a large number of coefficients when specifying a skedastic function model. This is desirable from a regulatory perspective, as guidance from regulatory agencies indicate that the number of covariates used in the statistical analysis of a RCT should be kept as small as possible. Furthermore, it facilitates the interpretability of the DTG and methodology, especially for non-statisticians. One can gain useful insights in a straightforward manner by comparing the $e_i^2$ and $s_i^2$, and considering the inferences on $\gamma_0$ and $\gamma_1$.

\subsection{Frequentist Properties of Weighted PROCOVA}
\label{sec:frequentist_properties}

We establish frequentist properties of the skedastic function model coefficient estimators $\hat{\gamma}$ and the regression coefficient estimators $\hat{\beta}$ under two perspectives. The first is the finite-sample perspective conditional on $\mathbf{U}$ and $\mathbf{V}$. The second is the asymptotic perspective, integrating over the distributions of the $u_i$ and $v_i$ and taking $N \rightarrow \infty$. For both cases our derivations are performed for the first iteration of Weighted PROCOVA, so as to facilitate the algebra. Under the finite-sample perspective, we calculate $\mathbb{E} \left ( \hat{\gamma} \mid \mathbf{U}, \mathbf{V} \right )$, $\mathbb{E} \left ( \hat{\beta} \mid \mathbf{U}, \mathbf{V} \right )$, and $\mathrm{Var} \left ( \hat{\gamma} \mid \mathbf{U}, \mathbf{V} \right )$ assuming only that the model in equation (\ref{eq:linear_model_observed_outcomes}), with $\epsilon_i \sim \mathrm{N} \left ( 0, \sigma_i^2 \right )$ independently, is correctly specified. We do not assume that the model in equation (\ref{eq:skedastic_function_model_residuals}) is correctly specified. Under the asymptotic perspective, we identify conditions on the predictors and the participant-level variances that are sufficient for the convergence of the expectation of $\hat{\gamma}$. The asymptotic covariance matrix of $\hat{\beta}$, which we use to characterize the variance reduction of Weighted PROCOVA compared to PROCOVA, is derived based on the work of \citet{romano_resurrecting_2016} and described in Section \ref{sec:frequentist_properties_variance_reduction}. All proofs are in Appendix \ref{sec:appendix_frequentist_proofs}.
 
The results on the finite-sample conditional expectations and variances are summarized below. 

\begin{prop}
\label{lem:conditional_expectations}
For the estimators $\hat{\gamma} = \left ( \mathbf{U}^{\mathsf{T}}\mathbf{U} \right )^{-1} \mathbf{U}^{\mathsf{T}} \mathrm{log} \left \{ \left ( y - \mathbf{H}y \right )^2 \right \}$,
\begin{align*}
\mathbb{E} \left ( \hat{\gamma}_0 \mid \mathbf{U}, \mathbf{V} \right ) &= 
\left [ N \displaystyle \sum_{k=1}^N \left \{ \mathrm{log} \left ( s_k^2 \right ) - \overline{\mathrm{log} \left ( s^2 \right )} \right \}^2 \right ]^{-1} \\
& \ \ \ \times \Bigg{(} \left [ \displaystyle \sum_{k=1}^N \left \{ \mathrm{log} \left ( s_k^2 \right ) \right \}^2 \right ] \displaystyle \sum_{i=1}^N \mathrm{log} \left \{ \left ( 1 - h_{ii} \right )^2 \sigma_i^2 + \displaystyle \sum_{k \neq i} h_{ik}^2 \sigma_k^2  \right \} \\ 
& \ \ \ \ \ \ \ - N\overline{\mathrm{log} \left ( s^2 \right )}  \displaystyle \sum_{i=1}^N \left [ \mathrm{log} \left ( s_i^2 \right ) \mathrm{log} \left \{ \left ( 1 - h_{ii} \right )^2 \sigma_i^2 + \displaystyle \sum_{k \neq i} h_{ik}^2 \sigma_k^2 \right \} \right ] \Bigg{)} - \left ( \gamma_{\mathrm{EM}} + \mathrm{log} 2 \right ),
\end{align*}

\begin{align*}
\mathbb{E} \left ( \hat{\gamma}_1 \mid \mathbf{U}, \mathbf{V} \right ) &= \left [ \displaystyle \sum_{k=1}^N \left \{ \mathrm{log} \left ( s_k^2 \right ) - \overline{\mathrm{log} \left ( s^2 \right )} \right \}^2 \right ]^{-1} \\
& \ \ \ \times \Bigg{(} \displaystyle \sum_{i=1}^N \left [ \left \{ \mathrm{log} \left ( s_i^2 \right ) - \overline{\mathrm{log} \left ( s^2 \right )} \right \} \mathrm{log} \left \{  \left ( 1 - h_{ii} \right )^2 \sigma_i^2 + \sum_{k \neq i} h_{ik}^2 \sigma_k^2 \right \} \right ] \Bigg{)},
\end{align*}
where $\gamma_{\mathrm{EM}} \approx 0.577$ is the Euler-Mascheroni constant.
\end{prop}

\begin{prop}
\label{lem:covariances}
Let $\left [ \left ( Z_1, \ldots, Z_N \right )^{\mathsf{T}} \mid \mathbf{U}, \mathbf{V} \right ]$ be Multivariate Normal with $\mathbb{E} \left ( Z_i \mid \mathbf{U}, \mathbf{V} \right ) = 0$, $\mathrm{Var} \left ( Z_i \mid \mathbf{U}, \mathbf{V} \right ) = 1$, and
\[
\mathrm{Cov} \left ( Z_i, Z_j \mid \mathbf{U}, \mathbf{V} \right ) = \frac{\displaystyle \sum_{k \neq i, j} h_{ik}h_{jk}\sigma_k^2 - h_{ij} \left ( 1 - h_{ii} \right )\sigma_i^2 - h_{ij} \left ( 1 - h_{jj} \right ) \sigma_j^2 }{\sqrt{\left ( 1 - h_{ii} \right )^2 \sigma_i^2 + \displaystyle \sum_{k \neq i} h_{ik}^2\sigma_k^2} \sqrt{ \left ( 1 - h_{jj} \right )^2 \sigma_j^2 + \displaystyle \sum_{k \neq j} h_{jk}^2 \sigma_k^2}}
\]
for distinct $i, j \in \left \{ 1, \ldots, N \right \}$. Then
\[
\mathrm{Var} \left ( \hat{\gamma}_0 \mid \mathbf{U}, \mathbf{V} \right ) = \frac{\left \{ \overline{\mathrm{log} \left ( s^2 \right )} \right \}^2 \mathrm{Var} \left [  \displaystyle \sum_{i=1}^N \left \{ 1 - \mathrm{log} \left ( s_i^2 \right ) \right \} \mathrm{log} \left ( Z_i^2 \right ) \mid \mathbf{U}, \mathbf{V} \right ]}{\left [ \displaystyle \sum_{i=1}^N \left \{ \mathrm{log} \left ( s_i^2 \right ) - \overline{\mathrm{log} \left ( s^2 \right )} \right \}^2 \right ]^{2}}
\]
\noindent and 
\[
\mathrm{Var} \left ( \hat{\gamma}_1 \mid \mathbf{U}, \mathbf{V} \right ) = \frac{\mathrm{Var} \left [ \displaystyle \sum_{i=1}^N \left \{ \mathrm{log} \left ( s_i^2 \right ) - \overline{\mathrm{log} \left ( s^2 \right )} \right \} \mathrm{log} \left ( Z_i^2 \right ) \mid \mathbf{U}, \mathbf{V} \right ]}{\left [ \displaystyle \sum_{i=1}^N \left \{ \mathrm{log} \left ( s_i^2 \right ) - \overline{\mathrm{log} \left ( s^2 \right )} \right \}^2 \right ]^{2}}.
\]
\end{prop}

\begin{theorem}
\label{thm:finite_sample_beta_expectation}
The estimator $\hat{\beta} = \left ( \mathbf{V}^{\mathsf{T}} \hat{\Omega}^{-1} \mathbf{V} \right )^{-1} \mathbf{V}^{\mathsf{T}}\hat{\Omega}^{-1}y$ is unbiased conditional on $\mathbf{U}$ and $\mathbf{V}$, i.e.,  $\mathbb{E} \left ( \hat{\beta} \mid \mathbf{U}, \mathbf{V} \right ) = \beta$.
\end{theorem}

\noindent A consequence of Theorem \ref{thm:finite_sample_beta_expectation} is super-population unbiasedness of $\hat{\beta}$.
\begin{corollary}
\label{lem:super-population_expectation}
The estimator $\hat{\beta} = \left ( \mathbf{V}^{\mathsf{T}} \hat{\Omega}^{-1} \mathbf{V} \right )^{-1} \mathbf{V}^{\mathsf{T}}\hat{\Omega}^{-1}y$ is unbiased unconditional on $\mathbf{U}$ and $\mathbf{V}$, i.e.,  $\mathbb{E} \left ( \hat{\beta} \right ) = \beta$.
\end{corollary}

The asymptotic expectations for $\hat{\gamma}$ are established under the following assumptions on the predictors $\mathbf{U}$ and $\mathbf{V}$, and the participant-level variances $\sigma_i^2$. These assumptions are motivated by considering the case in which a simple random sample of participants is taken from a super-population, and a completely randomized design is conducted for the participants. In this case, the $\left ( x_1, \sigma_1^2 \right ), \left ( x_2, \sigma_2^2 \right ), \ldots$ are independent and identically distributed in the super-population.

\begin{assumption}
\label{asmp:skedastic_function_predictors}
The predictors $\mathrm{log} \left ( s_1^2 \right ), \mathrm{log} \left ( s_2^2 \right ), \ldots$ for the super-population of participants are independent and identically distributed with finite mean and variance.
\end{assumption} 

\begin{assumption}
\label{asmp:predictors_sigma}
The predictors $v_i$ and the participant-level variances $\sigma_i^2$ in the super-population are distributed such that, for any finite sample of size $N > K+1$ of participants from the super-population, the $\mathrm{log} \left \{ \left ( 1 - h_{ii} \right )^2 \sigma_i^2 + \displaystyle \sum_{k \neq i} h_{ik}^2 \sigma_k^2 \right \}$ are identically distributed with finite mean and variance.
\end{assumption} 

\begin{assumption}
\label{asmp:predictors_infinite_limit}
For any finite sample of size $N > K+1$ of participants from the super-population, let $F_{H_{\sigma},N}$ denote the cumulative distribution function of the $\mathrm{log} \left \{ \left ( 1 - h_{11} \right )^2 \sigma_1^2 + \displaystyle \sum_{k \neq 1} h_{1k}^2 \sigma_k^2 \right \}$, $\ldots$, $\mathrm{log} \left \{  \left ( 1 - h_{NN} \right )^2 \sigma_N^2 + \displaystyle \sum_{k \neq N} h_{Nk}^2 \sigma_k^2 \right \}$. As $N \rightarrow \infty$, $F_{H_{\sigma},N}$ converges pointwise to a cumulative distribution function $F_{H_{\sigma}}$ with finite mean and variance.
\end{assumption} 
\noindent To simplify the presentation of the results, we let $\mathrm{log} \left ( S^2 \right )$ denote a random variable with the same distribution as the $\mathrm{log} \left (s_1^2 \right ), \mathrm{log} \left ( s_2^2 \right ), \ldots$ in the super-population and $\mathrm{log} \left ( H_{\sigma} \right )$ denote the logarithmic transformation of the random variable $H_{\sigma}$ that has cumulative distribution function $F_{H_{\sigma}}$.

\begin{prop}
\label{lem:consistency_expectation}
Under Assumptions \ref{asmp:skedastic_function_predictors} - \ref{asmp:predictors_infinite_limit}, as $N \rightarrow \infty$
\[
\mathbb{E} \left ( \hat{\gamma}_0 \right ) \rightarrow \frac{\mathbb{E} \left [ \left \{ \mathrm{log} \left ( S^2 \right ) \right \}^2 \right ] \mathbb{E} \left \{ \mathrm{log} \left ( H_{\sigma} \right ) \right \} - \mathbb{E} \left \{ \mathrm{log} \left ( S^2 \right ) \right \} \mathbb{E} \left \{ \mathrm{log} \left ( S^2 \right ) \mathrm{log} \left ( H_{\sigma} \right ) \right \}}{\mathrm{Var} \left \{ \mathrm{log} \left ( S^2 \right ) \right \}} - \left ( \gamma_{\mathrm{EM}} + \mathrm{log} 2 \right ),
\]
\[
\mathbb{E} \left ( \hat{\gamma}_1 \right ) \rightarrow \frac{\mathrm{Cov} \left \{ \mathrm{log} \left ( S^2 \right ), \mathrm{log} \left ( H_{\sigma} \right ) \right \}}{\mathrm{Var} \left \{ \mathrm{log} \left ( S^2 \right ) \right \}},
\]
\noindent where $\gamma_{\mathrm{EM}} \approx 0.577$ is the Euler-Mascheroni constant.
\end{prop}

We ultimately conclude under both the finite-sample and asymptotic perspectives that Weighted PROCOVA yields an unbiased treatment effect estimator. By calculating the standard error of the treatment effect estimator, e.g., via the bootstrap or the sandwich estimator \citep{white_1980}, one can also conduct hypothesis tests that control the Type I error rate and create confidence intervals that have the desired coverage rate. We illustrate these frequentist properties in the simulation studies of Section \ref{sec:simulation_studies}.

\subsection{Variance Reduction Under Weighted PROCOVA}
\label{sec:frequentist_properties_variance_reduction}

\citet{romano_resurrecting_2016} establish the asymptotic distribution of $\hat{\beta}$ under additional assumptions and conditions on the predictors and skedastic function. Their result facilitates an understanding and description of the factors driving the potential variance reduction of Weighted PROCOVA compared to PROCOVA. To describe the variance reduction, we will invoke the conditions in equations (3.10), (3.11), (3.12), and (3.13) of \citet[p.~5]{romano_resurrecting_2016}. Furthermore, we will assume that no residual from a fitted PROCOVA model will ever be exactly equal to zero, so as to eliminate the need for introducing the constant $\delta^2$ in equation (3.6) \citep[p.~4]{romano_resurrecting_2016}. In all our calculations we will have $v_i = \left ( 1, w_i, m_i \right )^{\mathsf{T}}$ and $u_i = \left (1, \mathrm{log} \left ( s_i^2 \right ) \right )^{\mathsf{T}}$. We focus on variance reduction in these calculations because the power boost from Weighted PROCOVA is a direct consequence of its reduction of the variance of the treatment effect estimator compared to PROCOVA. A demonstration of this general fact is provided by \citet{fisher_2023_power}. 

Under these assumptions, we have from Theorem 3.1 in \citep[p.~5]{romano_resurrecting_2016} that $\hat{\beta}$ converges in distribution to a Normal distribution centered at $\beta$ with asymptotic covariance matrix a function of the predictors in the mean model, the limit of the fitted skedastic function model, and the participant-level variances. More formally, we let $\mathcal{G}: \mathbb{R} \rightarrow \mathbb{R}_{>0}$ denote the limit of the fitted skedastic function model that has as its input $s_i^2$, and define
\begin{equation}
\label{eq:Omega_bread_matrix}
\Omega_{1/\mathcal{G} \left ( s^2 \right )} = \mathbb{E} \left \{ \frac{1}{\mathcal{G} \left (s_i^2 \right )} \begin{pmatrix} 1 & w_i & m_i \\ w_i & w_i & w_im_i \\ m_i & w_i m_i & m_i^2 \end{pmatrix} \right \}    
\end{equation}
and 
\begin{equation}
\label{eq:Omega_meat_matrix}
\Omega_{\sigma^2/\mathcal{G} \left ( s^2 \right )} = \mathbb{E} \left \{ \frac{\sigma_i^2}{\mathcal{G} \left (s_i^2 \right )^2} \begin{pmatrix} 1 & w_i & m_i \\ w_i & w_i & w_im_i \\ m_i & w_i m_i & m_i^2 \end{pmatrix} \right \}.    
\end{equation}
Then the asymptotic covariance matrix of $\hat{\beta}$ under Weighted PROCOVA is $\Omega_{1/\mathcal{G} \left ( s^2 \right )}^{-1} \Omega_{\sigma^2/\mathcal{G} \left ( s^2 \right )} \Omega_{1/\mathcal{G} \left ( s^2 \right )}^{-1}$. In addition, we have from Lemma 3.1 in \citep[p.~5]{romano_resurrecting_2016} that the asymptotic covariance matrix of the regression coefficient estimators under PROCOVA is
\begin{equation}
\label{eq:asymptotic_covariance_matrix_procova}
\left [ \mathbb{E} \left \{ \begin{pmatrix} 1 & w_i & m_i \\ w_i & w_i & w_im_i \\ m_i & w_im_i & m_i^2 \end{pmatrix} \right \} \right ]^{-1} \mathbb{E} \left \{ \sigma_i^2 \begin{pmatrix} 1 & w_i & m_i \\ w_i & w_i & w_im_i \\ m_i & w_im_i & m_i^2 \end{pmatrix} \right \} \left [ \mathbb{E} \left \{ \begin{pmatrix} 1 & w_i & m_i \\ w_i & w_i & w_im_i \\ m_i & w_im_i & m_i^2 \end{pmatrix} \right \} \right ]^{-1}.
\end{equation}
The variance reduction of Weighted PROCOVA to PROCOVA is thus calculated as the ratio of the $(2,2)$ entry in the matrix $\Omega_{1/\mathcal{G} \left ( s^2 \right )}^{-1} \Omega_{\sigma^2/\mathcal{G} \left ( s^2 \right )} \Omega_{1/\mathcal{G} \left ( s^2 \right )}^{-1}$ with the corresponding entry of the matrix in equation (\ref{eq:asymptotic_covariance_matrix_procova}). 

A simplified expression for the variance reduction follows by means of the following three assumptions under the super-population perspective.

\begin{assumption}
\label{asmp:uncorrelated_1}
The transformed predictors $1/\mathcal{G} \left ( s_i^2 \right )$ are uncorrelated with the $m_i, m_i^2, w_i$, and $w_im_i$.
\end{assumption} 

\begin{assumption}
\label{asmp:uncorrelated_2}
The participant-level variances $\sigma_i^2$ are uncorrelated with the $m_i, m_i^2, w_i$, and $w_im_i$.
\end{assumption} 

\begin{assumption}
\label{asmp:uncorrelated_3}
The ratios of the participant-level variances to the square of the transformed predictors, i.e., the $\sigma_i^2 \left \{ \mathcal{G} \left ( s_i^2 \right ) \right \}^{-2}$ are uncorrelated with the $m_i, m_i^2, w_i$, and $w_im_i$.
\end{assumption} 
\noindent These assumptions can be also motivated by recognition of the physical act of randomization in a RCT, and by considering the case in which the prognostic scores are independent of the personalized precisions and participant-level variances. Under Assumptions \ref{asmp:uncorrelated_1} - \ref{asmp:uncorrelated_3}, the expectations in equations (\ref{eq:Omega_bread_matrix}), (\ref{eq:Omega_meat_matrix}), and \ref{eq:asymptotic_covariance_matrix_procova}) can be simplified to yield the expression
\begin{equation}
\label{eq:variance_reduction_eta}
1 - \frac{\mathbb{E} \left [ \sigma_i^2\left \{ \mathcal{G} \left ( s_i^2 \right ) \right \}^{-2} \right ]}{\mathbb{E} \left \{ \mathcal{G} \left ( s_i \right )^{-2} \right \} \mathbb{E} \left ( \sigma_i^2 \right )}.
\end{equation}
Finally, when we have the limiting case of $\mathcal{G} \left ( s_i^2 \right ) = \mathrm{exp} \left \{ \gamma_0 + \gamma_1 \mathrm{log} \left ( s_i^2 \right ) \right \}$, then we have that the variance reduction is a function of $\gamma_1$. We demonstrate this result via simulation in Section \ref{sec:simulation_studies_results_boost}.

\section{Simulation Studies}
\label{sec:simulation_studies}

\subsection{Data Generation Mechanisms and Evaluation Metrics}
\label{sec:generative_models}

We conduct two sets of simulation studies to compare and contrast the frequentist properties of Weighted PROCOVA with PROCOVA and the unadjusted analysis. The first set involves correctly specified mean and variance component models. The second set introduces a predictor in the data generation mechanism that is not in the statistical analysis to illuminate the impact of skedastic function model misspecification on the properties. In all of our simulations we assume no covariate drifts or changes in standard of care.

The observed outcomes are generated according to
\begin{equation}
\label{eq:mean_model_simulation}
y_i = \beta_0 + \beta_1 w_i + \beta_2 m_i + \epsilon_i,
\end{equation}
where $\epsilon_i \sim \mathrm{N} \left ( 0, \sigma_i^2 \right )$ independently, and the $\sigma_i^2$ are generated according to
\begin{equation}
\label{eq:skedastic_function_simulation}
\mathrm{log} \left ( \sigma_i^2 \right ) = \gamma_0 + \gamma_1 \mathrm{log} \left ( s_i^2 \right ) + \gamma_2 u_{i,2} + \zeta_i,
\end{equation}
where $u_{i,2}$ is an additional factor that can change the level of heteroskedasticity but is not included in any Weighted PROCOVA analysis, and $\zeta_i \sim \mathrm{N} \left ( 0, \psi^2 \right )$ independently. The $u_{i,2}$ are introduced to evaluate the effects of skedastic function model misspecification on the frequentist properties of Weighted PROCOVA. For the first set of simulation studies we fix $\gamma_2 = 0$ to consider a correctly specified skedastic function model, and for the second set we consider a range of non-zero $\gamma_2$ values to consider a misspecified skedastic function model in terms of the omission of a predictor variable. The $\zeta_i$ are unobservable sources of variation for the participant-level variances. We generate $m_i \sim \mathrm{N} \left ( 0, \tau_1^2 \right )$, $\mathrm{log} \left ( s_i^2 \right ) \sim \mathrm{N} \left ( 0, \tau_2^2 \right )$, and $u_{i,2} \sim \mathrm{N} \left ( 0, \tau_3^2 \right )$ independently. 

Each set of simulation studies consists of three types of heteroskedasticity scenarios. In the first scenario the total level of heteroskedasticity, defined as $\mathrm{Var} \left \{ \mathrm{log} \left ( \sigma_i^2 \right ) \right \}$, is fixed across all combinations of simulation parameters. The quantity $\mathrm{Var} \left \{ \mathrm{log} \left ( \sigma_i^2 \right ) \right \}$ is the sum of the variation explained by the predictors of the skedastic function that are contained in $x_i$, and the variation due to unobserved sources. More formally, in this scenario $\gamma_1^2\tau_2^2 + \psi^2$ remains constant across the combinations of simulation parameters, or alternatively $\psi = \sqrt{\mathrm{Var} \left \{ \mathrm{log} \left ( \sigma_i^2 \right ) \right \} - \gamma_1^2 \tau_2^2}$ for a fixed value of $\mathrm{Var} \left \{ \mathrm{log} \left ( \sigma_i^2 \right ) \right \}$ and values of $\gamma_1$ and $\tau_2^2$ such that $\gamma_1^2 \tau_2^2 < \mathrm{Var} \left \{ \mathrm{log} \left ( \sigma_i^2 \right ) \right \}$. As $\mathrm{log} \left ( s_i^2 \right )$ explains more of the variation in the $\mathrm{log} \left ( \sigma_i^2 \right )$ the variation of the $\zeta_i$ decreases. This scenario enables us to characterize the quality of $\mathrm{log} \left ( s_i^2 \right )$ for Weighted PROCOVA in terms of the correlation between $\mathrm{log} \left ( s_i^2 \right )$ and $\mathrm{log} \left ( \sigma_i^2 \right )$. The second scenario differs from the first in that the participants' variances are a deterministic function of the predictors of the skedastic function, i.e., $\psi = 0$ across combinations of simulation parameters. When $\gamma_2 = 0$ all variation in the $\mathrm{log} \left ( \sigma_i^2 \right )$ will arise from that of $\mathrm{log} \left ( s_i^2 \right )$, and $\mathrm{Var} \left ( \mathrm{log} \left ( \sigma_i^2 \right ) \right \} \propto \tau_2^2$ with proportionality constant $\gamma_1^2$. If $\gamma_1 < 1$ then $\mathrm{Var} \left \{ \mathrm{log} \left ( \sigma_i^2 \right ) \right \} < \tau_2^2$ and we may expect that Weighted PROCOVA will not have desirable variance reduction or power boost properties, whereas if $\gamma_1 > 1$ then $\mathrm{Var} \left \{ \mathrm{log} \left ( \sigma_i^2 \right ) \right \} > \tau_2^2$ and we would expect it to perform well. Although the skedastic function is deterministic, our study of the second scenario still enables us to learn how the change in the relationship between the $\mathrm{log} \left ( s_i^2 \right )$ and $\mathrm{log} \left ( \sigma_i^2 \right )$ as characterized by $\gamma_1$ affects the frequentist properties of Weighted PROCOVA. This scenario was also considered by \citet[p.~9]{romano_resurrecting_2016} and can arise in practice when subpopulations of participants with identical covariates also have identical outcome variances. The final scenario is more general than the first two in that the participants' variances are not a deterministic function of their covariates, and the total level of heteroskedasticity is not fixed across all combinations of simulation parameters. Specifically, $\psi^2 > 0$ is fixed and $\mathrm{Var} \left \{ \mathrm{log} \left ( \sigma_i^2 \right ) \right \}$ can change across the combinations of simulation parameters. This scenario helps us understand how the combination of unobserved sources of variation and changes in the correlation between the $\mathrm{log} \left ( s_i^2 \right )$ and $\mathrm{log} \left ( \sigma_i^2 \right )$ affect the performance of Weighted PROCOVA. Our consideration of all these scenarios ultimately enables us to map out the statistical regimes for Weighted PROCOVA in terms of its frequentist properties.

The analysis model that we consider throughout has the same mean model specification as in equation (\ref{eq:mean_model_simulation}), but only includes $\mathrm{log} \left ( s_i^2 \right )$ as the sole predictor in the skedastic function model and omits $u_{i,2}$ from the analysis. For each setting we simulate $10^4$ datasets to evaluate the biases of $\hat{\beta}_1$, Type I error rates of the tests for $H_0: \beta_1 = 0$, and coverage rates of the confidence intervals for $\beta_1$ for the three methods. For the power evaluations we set the true value of $\beta_1$ to $0.4$, and for each combination of simulation parameters we select $N$ such that PROCOVA has $80 \%$ power for rejecting the null. We evaluate the percentage variance reduction of Weighted PROCOVA compared to PROCOVA for the case of $\beta_1 = 0$, as we consider additive treatment effects throughout and the actual value of the treatment effect should not affect the variance reduction.

A simulation setting is defined by the combination of values for $N, \beta_0, \beta_1, \beta_2, \gamma_0, \gamma_1, \gamma_2, \tau_1^2, \tau_2^2, \tau_3^2$, and $\psi^2$. We consider ranges of values for $N, \beta_1, \beta_2, \gamma_0, \gamma_1, \tau_2^2$, and $\psi^2$. For the null case of $\beta_1 = 0$, the range of RCT sample sizes are $N = 50, 100, 300, 500, 1000$. In all simulation settings the number of treated participants is $N_1 = N/2$. We fix $\beta_0 = 0$ and $\tau_1^2 = \tau_3^2 = 1$ throughout. The values for $\gamma_0, \gamma_1, \gamma_2, \tau_2^2$, and $\psi^2$ in our simulation studies are summarized in Table \ref{tab:simulation_parameter_settings}. The range of $\gamma_1$ values correspond to $R^2$ values for the skedastic function model fit being less that $0.1 \sim 0.2$ on average. 

\begin{table} 
\centering
\begin{tabular}{| c | c | c | c | }
\hline
Parameter & Setting 1 & Setting 2 & Setting 3 \\
\hline
$N$ & $50, 100, 300,$ & $50, 100, 300,$ & $50, 100, 300,$ \\
(null case) & $500, 1000$ & $500, 1000$ & $500, 1000$ \\
\hline
$N$ & $250$ & $210, 220, 240$ & $350, 375, 400$ \\
(alternative case) & & $250, 280, 325$ & $425, 475, 525$ \\
\hline
$\beta_1$ & $0, 0.4$ & $0, 0.4$ & $0, 0.4$ \\
\hline
$\beta_2$ & $0.2, 0.3, 0.4,$ & $0.2, 0.3, 0.4, $ & $0.2, 0.3, 0.4,$ \\
 & $0.5, 0.6$ & $0.5, 0.6$ & $0.5, 0.6$ \\
\hline
$\gamma_0$ & $-1.75$ & $0$ & $0$ \\
\hline
$\gamma_1$ & $0.4, 0.6,$ & $0.4, 0.6,$ & $0.4, 0.6,$ \\
& $0.8, 1, 1.2, 1.4$ & $0.8, 1, 1.2, 1.4$ & $0.8, 1, 1.2, 1.4$ \\
\hline
$\gamma_2$ & $0, 0.25, 0.5, 0.75,$ & $0, 0.25, 0.5, 0.75,$ & $0, 0.25, 0.5, 0.75,$ \\
 & $1, 1.25, 1.5$ & $1, 1.25, 1.5$ & $1, 1.25, 1.5$ \\
\hline
$\tau_2^2$ & $0.5$ & $0.5$ & $0.5$ \\
\hline
$\psi^2$ & $4 - \gamma_1^2\tau_2^2$ & $0$ & $1$ \\
\hline
\end{tabular}
\caption{Values of the parameters considered in our two sets of simulation studies. In the alternative case of $\beta_1 = 0.4$, the value of $N$ depends on the other parameter values as it is selected so that PROCOVA has power $80\%$ for rejecting the null. The first set of experiments has $\gamma_2 = 0$, and the second set has $\gamma_2 \neq 0$.}
\label{tab:simulation_parameter_settings}
\end{table}

\subsection{Bias, Type I Error Rate, and Confidence Interval Coverage}
\label{sec:simulation_studies_results_freq}

Table \ref{tab:freq_table_experiments_1_2} summarizes the average biases, Type I error rates, and confidence interval coverage rates for Weighted PROCOVA in the case of $\beta_1 = 0$ across the values of the other simulation parameters. We observe that Weighted PROCOVA yields unbiased estimators of the treatment effects, and controls the Type I error rates and confidence interval coverage rates. Tables \ref{tab:main_effects_experiment_1} and \ref{tab:main_effects_experiment_2} contain inferences on the effects of the simulation parameters $\beta_2, \gamma_1, \gamma_2$, and $N$ on these metrics. Specifically, we provide both the estimated main effects of these parameters and the p-values for testing whether they have any effects on the metrics. These inferences are obtained via linear regression on the results from the simulation studies. We see that all of the main effects for these parameters are estimated to be zero, and that none of them are statistically significant. This helps to indicate that for these ranges of $\beta_2, \gamma_1, \gamma_2$, and $N$, these parameters don't affect the frequentist properties of Weighted PROCOVA.

\begin{table}[H]
\centering
\begin{tabular}{| c | c | c | c | c |}
\hline
Value of $\gamma_2$ & Setting & Type I Error Rate & Confidence Interval Coverage & Bias \\
\hline
$\gamma_2 = 0$ & Setting 1 & $0.0506$ & $0.9494$ & $10^{-4}$ \\
\hline
$\gamma_2 = 0$ & Setting 2 & $0.0558$ & $0.9442$ & $0$ \\
\hline
$\gamma_2 = 0$ & Setting 3 & $0.0545$ & $0.9455$ & $0$ \\
\hline
$\gamma_2 \neq 0$ & Setting 1 & $0.0501$ & $0.9499$ & $0$ \\
\hline
$\gamma_2 \neq 0$ & Setting 2 & $0.0556$ & $0.9444$ & $-10^{-4}$ \\
\hline
$\gamma_2 \neq 0$ & Setting 3 & $0.0546$ & $0.9454$ & $10^{-4}$ \\
\hline
\end{tabular}
\caption{Summary of Type I error rates, confidence interval coverages, and biases for Weighted PROCOVA, averaged over the simulation parameter values in Table \ref{tab:simulation_parameter_settings}. The first set of simulation studies correspond to the rows with $\gamma_2 = 0$, and the second set correspond to the rows with $\gamma_2 \neq 0$.}
\label{tab:freq_table_experiments_1_2}
\end{table}

\begin{table}[H]
\centering
\begin{tabular}{| c | c | c | c | c |}
\hline
Parameter & Setting & Main Effect ($p$-value) & Main Effect ($p$-value) & Main Effect ($p$-value) \\
& & Type I Error Rate & Coverage Rate & Bias \\
\hline
$\beta_m$ & Setting 1 & $-0.001$ ($0.394$) & $0.001$ ($0.394$) & $-0.001$ ($0.5$)\\
\hline
$\gamma_1$ & Setting 1 & $0.004$ ($<0.0001$) & $-0.004$ ($<0.0001$) & $0$ ($0.87$)\\
\hline
$N$ & Setting 1 & $0$ ($0.235$) & $0$ ($0.235$) & $0$ ($0.884$)\\
\hline
$\beta_m$ & Setting 2 & $0.002$ ($0.379$) & $-0.002$ ($0.379$) & $0$ ($0.671$)\\
\hline
$\gamma_1$ & Setting 2 & $0.003$ ($<0.0001$) & $-0.003$ ($<0.0001$) & $0$ ($0.848$)\\
\hline
$N$ & Setting 2 & $0$ ($<0.0001$) & $0$ ($<0.0001$) & $0$ ($0.975$)\\
\hline
$\beta_m$ & Setting 3 & $-0.002$ ($0.36$) & $0.002$ ($0.36$) & $0$ ($0.991$)\\
\hline
$\gamma_1$ & Setting 3 & $0.003$ ($<0.0001$) & $-0.003$ ($<0.0001$) & $0$ ($0.864$)\\
\hline
$N$ & Setting 3 & $0$ ($<0.0001$) & $0$ ($<0.0001$) & $0$ ($0.842$)\\
\hline
\end{tabular}
\caption{Main effects and $p$-values of the effects of the parameters $\beta_2, \gamma_1$, and $N$ on the Type I error rates, confidence interval coverage rates, and biases of Weighted PROCOVA in the case of $\gamma_2 = 0$.}
\label{tab:main_effects_experiment_1}
\end{table}

\begin{table}[H]
\centering
\begin{tabular}{| c | c | c | c | c |}
\hline
Parameter & Setting & Main Effect ($p$-value) & Main Effect ($p$-value) & Main Effect ($p$-value) \\
& & Type I Error Rate & Coverage Rate & Bias \\
\hline
$\gamma_2$ & Setting 1 & $-0.003$ ($<0.0001$) & $0.003$ ($<0.0001$) & $-0.001$ ($0.259$)\\
\hline
$\gamma_1$ & Setting 1 & $0.004$ ($<0.0001$) & $-0.004$ ($<0.0001$) & $0$ ($0.403$)\\
\hline
$N$ & Setting 1 & $0$ ($0.02$) & $0$ ($0.02$) & $0$ ($0.547$)\\
\hline
$\gamma_2$ & Setting 2 & $-0.001$ ($0.279$) & $0.001$ ($0.279$) & $0$ ($0.774$)\\
\hline
$\gamma_1$ & Setting 2 & $0.003$ ($0.005$) & $-0.003$ ($0.005$) & $-0.001$ ($0.031$)\\
\hline
$N$ & Setting 2 & $0$ ($<0.0001$) & $0$ ($<0.0001$) & $0$ ($0.281$)\\
\hline
$\gamma_2$ & Setting 3 & $-0.002$ ($0.016$) & $0.002$ ($0.016$) & $0$ ($0.86$)\\
\hline
$\gamma_1$ & Setting 3 & $0.003$ ($0.003$) & $-0.003$ ($0.003$) & $0$ ($0.503$)\\
\hline
$N$ & Setting 3 & 0 ($<0.0001$) & $0$ ($<0.0001$) & $0$ ($0.476$)\\
\hline
\end{tabular}
\caption{Main effects and $p$-values of the effects of the parameters $\beta_2, \gamma_1, \gamma_2$, and $N$ on the Type I error rates, confidence interval coverage rates, and biases of Weighted PROCOVA in the case of $\gamma_2 \neq 0$.}
\label{tab:main_effects_experiment_2}
\end{table}

\subsection{Variance Reduction and Power Boost}
\label{sec:simulation_studies_results_boost}

Figures \ref{fig:pct_varred_fitted_exp1setting1} to \ref{fig:pct_varred_fitted_exp3setting3} summarize the percentage variance reductions of Weighted PROCOVA compared to PROCOVA for the two sets of simulations in which $\gamma_2 = 0$ and $\gamma_2 \neq 0$, where $\beta_1 = 0$ throughout. Variance reduction is significantly affected by $\gamma_1$, and increases as a function of $\gamma_1$. As $N$ increases the variation in variance reduction decreases, and the probability of an inflation of the variance compared to PROCOVA decreases. In addition, as $\gamma_2$ increases the variation in the variance reduction increases, and the probability of an inflation of the variance compared to PROCOVA increases. The $N$ and $\gamma_2$ parameters have an interaction effect on variance reduction (Figures \ref{fig:pct_varred_fitted_exp3setting1} to \ref{fig:pct_varred_fitted_exp3setting3}). The values of $\beta_2$ do not affect variance reduction. The second and third heteroskedasticity scenarios exhibit less variation in variance reduction compared to the first scenario for the cases of $\gamma_2 = 0$ and $\gamma_2 \neq 0$, respectively. 

Tables \ref{tab:power_table_1} and \ref{tab:power_table_2} summarize the power boosts of Weighted PROCOVA over PROCOVA for the two sets of simulations in which $\gamma_2 = 0$ and $\gamma_2 \neq 0$. This metric is driven by variance reduction, and so the relationships between the simulation parameters and power boost are similar to those for the case of variance reduction. The percentage variance reductions and power boosts of Weighted PROCOVA are of practical significance, as they can lead to the reduction of control arm sample sizes, and accelerate the timelines, of RCTs.

The dependence of Weighted PROCOVA's performance metrics on $\gamma_1$ corresponds to the fact that as $\gamma_1$ increases the $s_i^2$ explain more of the heteroskedasticity in the outcomes. This also follows from the fact that $\gamma_1$ is related to the expected coefficient of determination $R^2$ value for the skedastic function model fit. The values of $\gamma_1$ in our simulation studies correspond to expected $R^2$ values ranging from $0.1$ to $0.2$ for the skedastic function model fit. The percentage of variance reduction offered by Weighted PROCOVA can reach up to $50\%$, even when the skedastic function model is misspecified, so long as the personalized precisions explain the heteroskedasticity in the RCT. For small $\gamma_1$ the amount of heteroskedasticity would be negligible, so it would be preferable to use PROCOVA. 

\begin{figure}[H]
\centering
\includegraphics[scale=0.2]{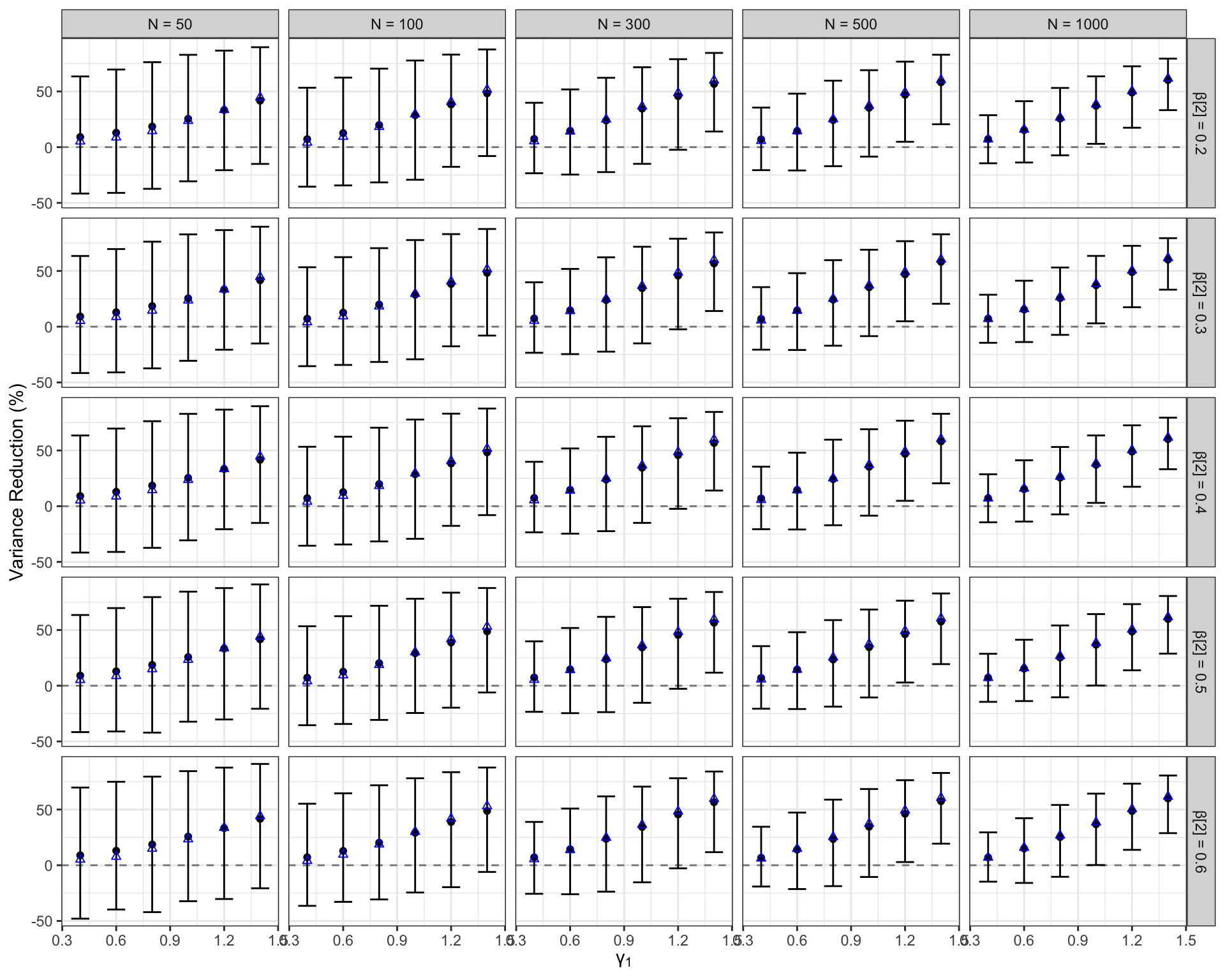}
\caption{Percentage variance reduction for Weighted PROCOVA compared to PROCOVA under the first heteroskedasticity scenario, with $\gamma_2 = 0$ and $\beta_1 = 0$. The triangle indicates the median of the percentage variance reduction, and the dot indicates the mean of the percentage variance reduction, across the simulated datasets.}
\label{fig:pct_varred_fitted_exp1setting1}
\end{figure}

\begin{figure}[H]
\centering
\includegraphics[scale=0.2]{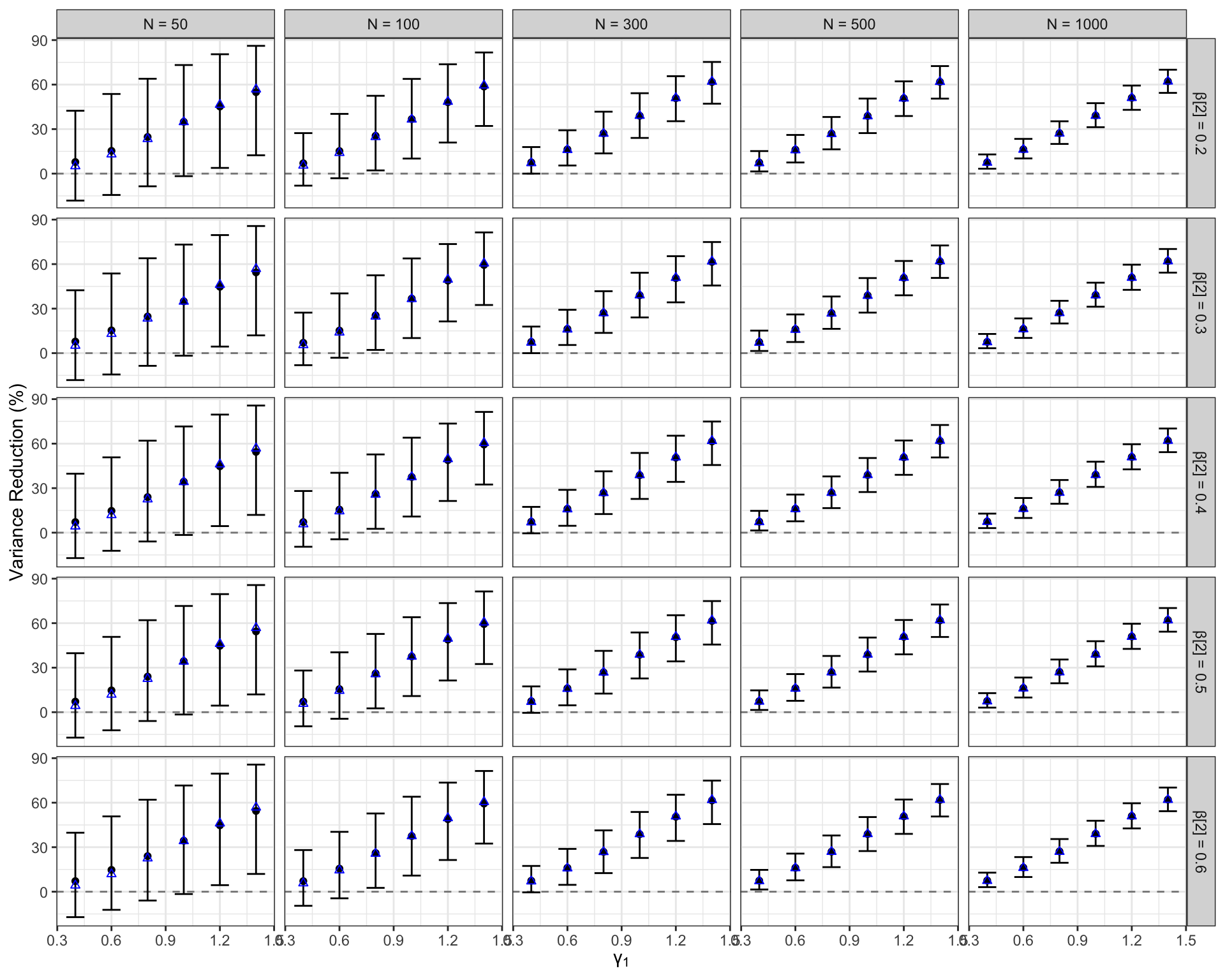}
\caption{Percentage variance reduction for Weighted PROCOVA compared to PROCOVA under the second heteroskedasticity scenario, with $\gamma_2 = 0$ and $\beta_1 = 0$. The triangle indicates the median of the percentage variance reduction, and the dot indicates the mean of the percentage variance reduction, across the simulated datasets.}
\label{fig:pct_varred_fitted_exp1setting2}
\end{figure}

\begin{figure}[H]
\centering
\includegraphics[scale=0.2]{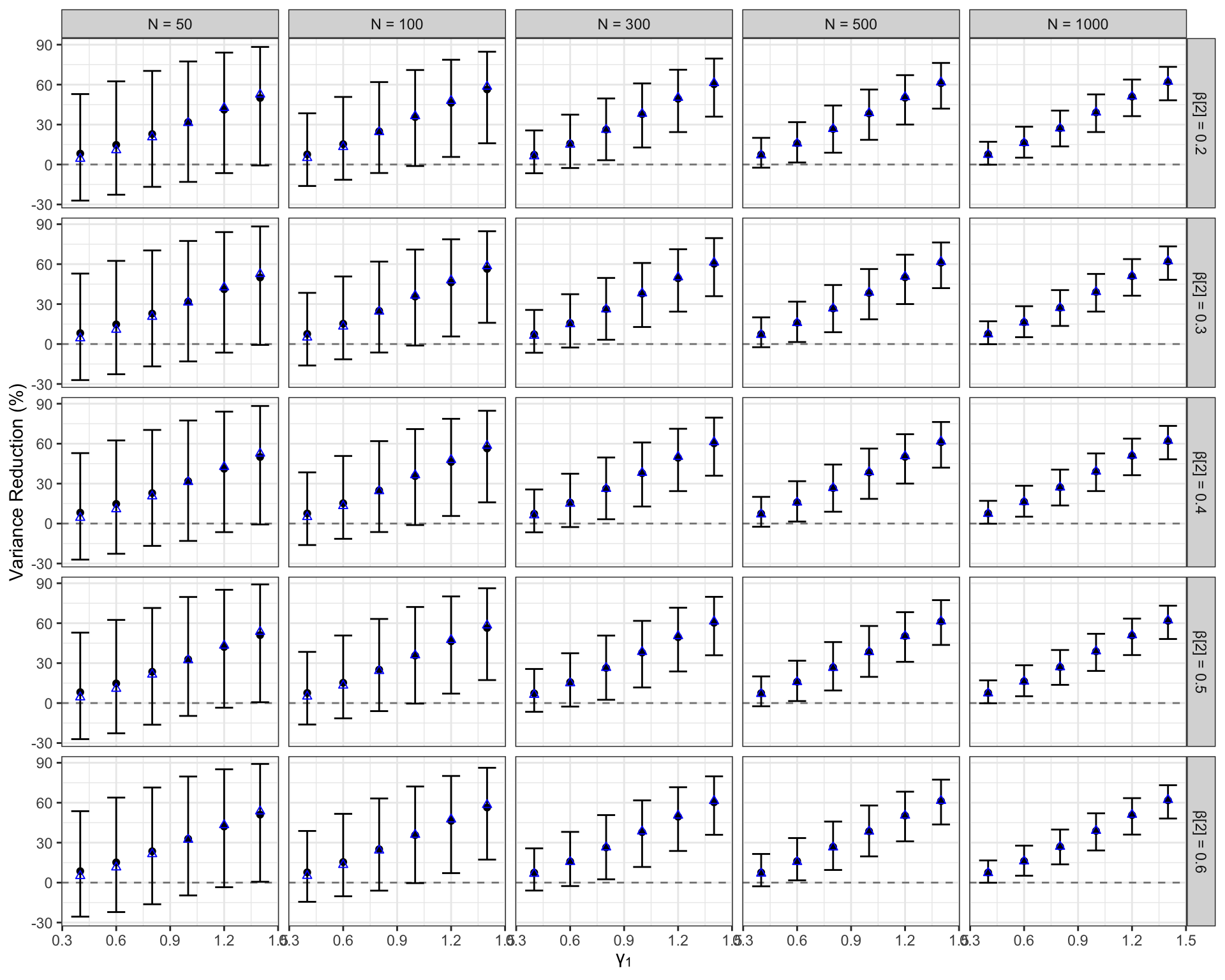}
\caption{Percentage variance reduction for Weighted PROCOVA compared to PROCOVA under the third heteroskedasticity scenario, with $\gamma_2 = 0$ and $\beta_1 = 0$. The triangle indicates the median of the percentage variance reduction, and the dot indicates the mean of the percentage variance reduction, across the simulated datasets.}
\label{fig:pct_varred_fitted_exp1setting3}
\end{figure}

\begin{figure}[H]
\centering
\includegraphics[scale=0.2]{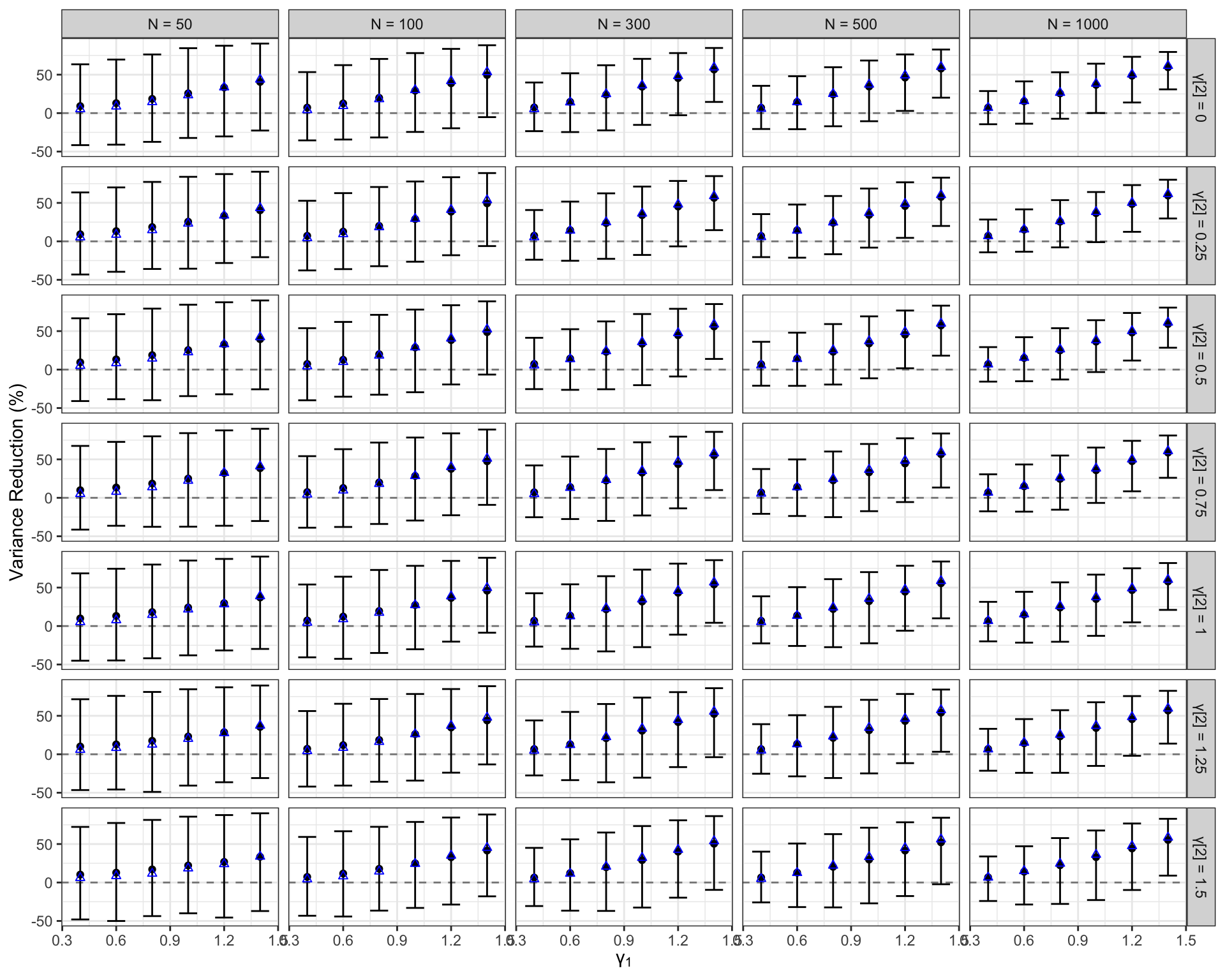}
\caption{Percentage variance reduction for Weighted PROCOVA compared to PROCOVA under the first heteroskedasticity scenario, with $\gamma_2 \neq 0$ and $\beta_1 = 0$. The triangle indicates the median of the percentage variance reduction, and the dot indicates the mean of the percentage variance reduction, across the simulated datasets.}
\label{fig:pct_varred_fitted_exp3setting1}
\end{figure}

\begin{figure}[H]
\centering
\includegraphics[scale=0.2]{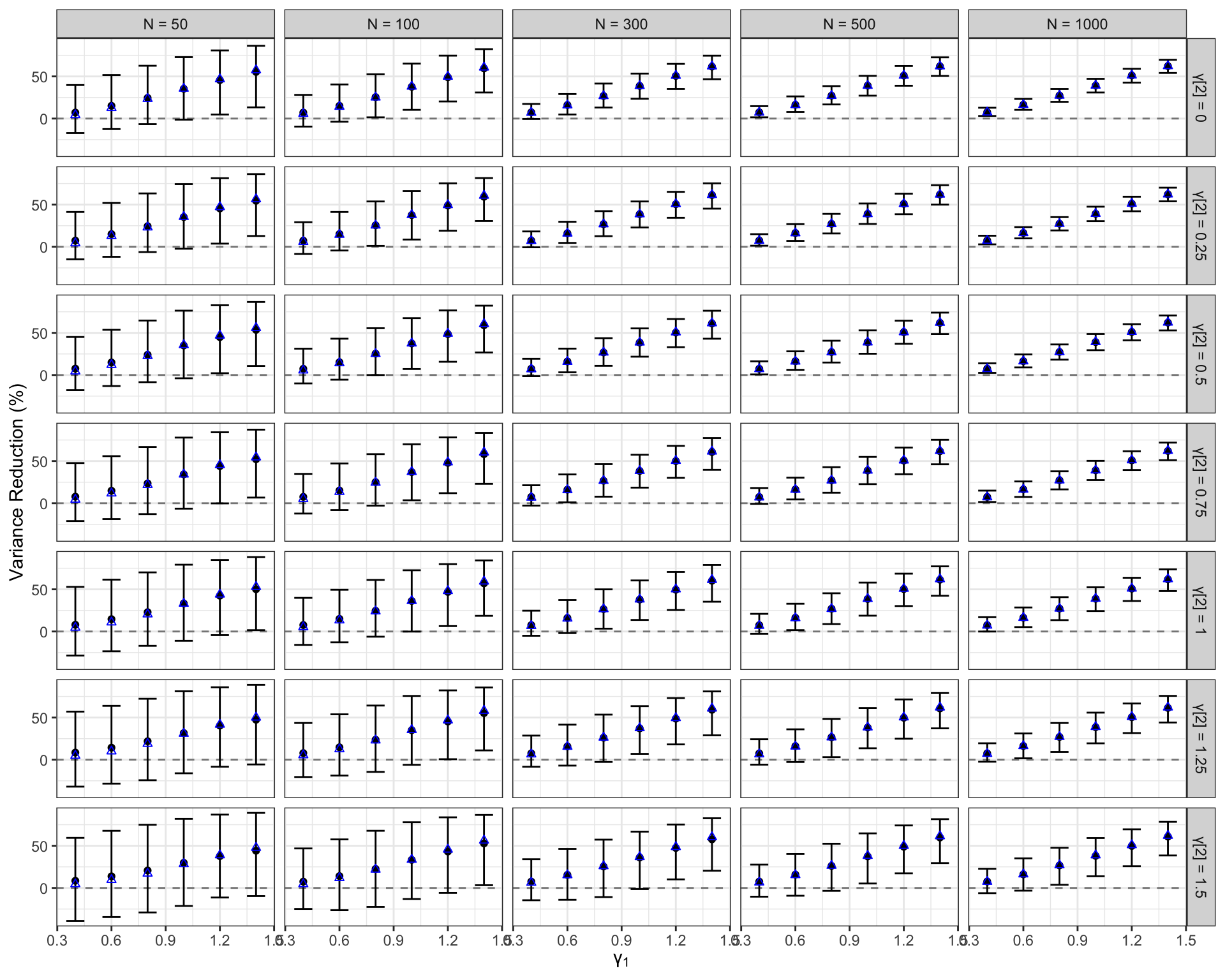}
\caption{Percentage variance reduction for Weighted PROCOVA compared to PROCOVA under the second heteroskedasticity scenario, with $\gamma_2 \neq 0$ and $\beta_1 = 0$. The triangle indicates the median of the percentage variance reduction, and the dot indicates the mean of the percentage variance reduction, across the simulated datasets.}
\label{fig:pct_varred_fitted_exp3setting2}
\end{figure}

\begin{figure}[H]
\centering
\includegraphics[scale=0.2]{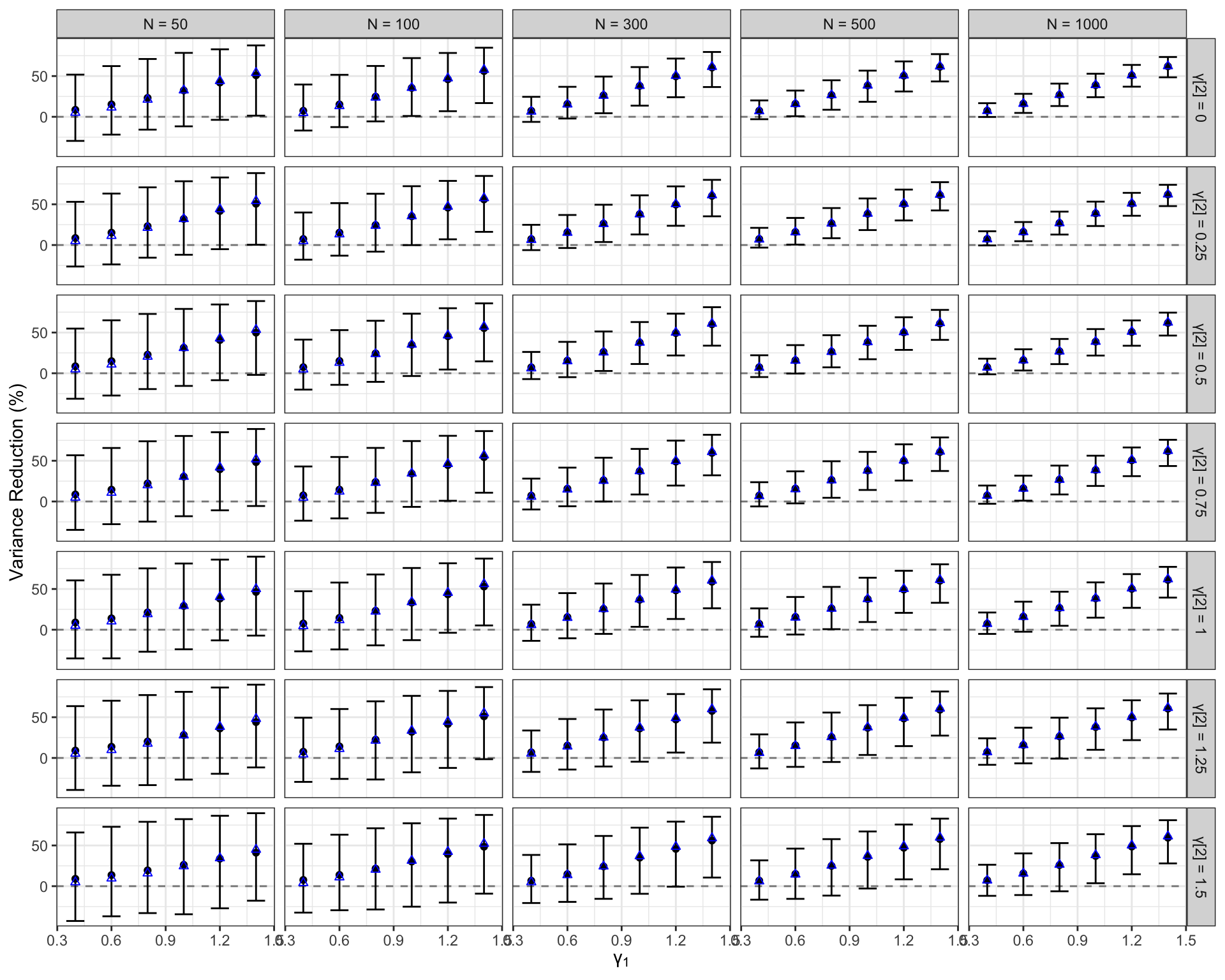}
\caption{Percentage variance reduction for Weighted PROCOVA compared to PROCOVA under the third heteroskedasticity scenario, with $\gamma_2 \neq 0$ and $\beta_1 = 0$. The triangle indicates the median of the percentage variance reduction, and the dot indicates the mean of the percentage variance reduction, across the simulated datasets.}
\label{fig:pct_varred_fitted_exp3setting3}
\end{figure}

\begin{sidewaystable}[!ph] 
\begin{tabular}{|c|c|l|l|l|l|l|l|l|}
\hline
$\beta_m$ & Setting & $\gamma_1=0.4$ & $\gamma_1=0.6$ & $\gamma_1=0.8$ & $\gamma_1=1$ & $\gamma_1=1.2$ & $\gamma_1=1.4$ \\
\hline
0.2 & 1 & 0.843 (0.03) & 0.863 (0.05) & 0.895 (0.08) & 0.921 (0.1) & 0.95 (0.13) & 0.979 (0.16)\\
\hline
0.3 & 1 & 0.843 (0.03) & 0.863 (0.05) & 0.895 (0.08) & 0.921 (0.1) & 0.95 (0.13) & 0.979 (0.16)\\
\hline
0.4 & 1 & 0.843 (0.03) & 0.863 (0.05) & 0.895 (0.08) & 0.921 (0.1) & 0.95 (0.13) & 0.979 (0.16)\\
\hline
0.5 & 1 & 0.843 (0.03) & 0.863 (0.05) & 0.895 (0.08) & 0.921 (0.1) & 0.95 (0.13) & 0.979 (0.16)\\
\hline
0.6 & 1 & 0.843 (0.03) & 0.863 (0.05) & 0.895 (0.08) & 0.921 (0.1) & 0.95 (0.13) & 0.979 (0.16)\\
\hline
0.2 & 2 & 0.838 (0.03) & 0.857 (0.06) & 0.911 (0.11) & 0.942 (0.15) & 0.977 (0.19) & 0.995 (0.18)\\
\hline
0.3 & 2 & 0.838 (0.03) & 0.857 (0.06) & 0.911 (0.11) & 0.942 (0.15) & 0.977 (0.19) & 0.995 (0.18)\\
\hline
0.4 & 2 & 0.838 (0.03) & 0.857 (0.06) & 0.911 (0.11) & 0.942 (0.15) & 0.977 (0.19) & 0.995 (0.18)\\
\hline
0.5 & 2 & 0.838 (0.03) & 0.857 (0.06) & 0.908 (0.1) & 0.943 (0.14) & 0.977 (0.17) & 0.993 (0.19)\\
\hline
0.6 & 2 & 0.832 (0.03) & 0.866 (0.06) & 0.908 (0.1) & 0.943 (0.14) & 0.977 (0.17) & 0.993 (0.19)\\
\hline
0.2 & 3 & 0.83 (0.02) & 0.882 (0.06) & 0.914 (0.1) & 0.945 (0.13) & 0.977 (0.17) & 0.992 (0.19)\\
\hline
0.3 & 3 & 0.83 (0.02) & 0.882 (0.06) & 0.914 (0.1) & 0.945 (0.13) & 0.977 (0.17) & 0.992 (0.19)\\
\hline
0.4 & 3 & 0.83 (0.02) & 0.882 (0.06) & 0.914 (0.1) & 0.945 (0.13) & 0.977 (0.17) & 0.992 (0.19)\\
\hline
0.5 & 3 & 0.83 (0.02) & 0.882 (0.06) & 0.914 (0.1) & 0.945 (0.13) & 0.977 (0.17) & 0.992 (0.19)\\
\hline
0.6 & 3 & 0.83 (0.02) & 0.882 (0.06) & 0.914 (0.1) & 0.945 (0.13) & 0.977 (0.17) & 0.992 (0.19)\\
\hline
\end{tabular}
\caption{Power and power boosts of Weighted PROCOVA over PROCOVA in the case that $\gamma_2 = 0$. Throughout all simulations the power of PROCOVA was $0.8$, and the power boosts of Weighted PROCOVA are in parentheses.}
\label{tab:power_table_1}
\end{sidewaystable}

\begin{sidewaystable}[!ph] 
\begin{tabular}{|c|c|l|l|l|l|l|l|}
\hline
$\gamma_2$ & Setting & $\gamma_1=0.4$ & $\gamma_1=0.6$ & $\gamma_1=0.8$ & $\gamma_1=1$ & $\gamma_1=1.2$ & $\gamma_1=1.4$ \\
\hline
0.00 & 1 & 0.828 (0.03) & 0.797 (0.05) & 0.849 (0.1) & 0.873 (0.13) & 0.919 (0.18) & 0.972 (0.16)\\
\hline
0.25 & 1 & 0.822 (0.03) & 0.862 (0.05) & 0.84 (0.1) & 0.917 (0.1) & 0.953 (0.14) & 0.967 (0.17)\\
\hline
0.50 & 1 & 0.798 (0.03) & 0.831 (0.04) & 0.869 (0.09) & 0.899 (0.11) & 0.943 (0.16) & 0.963 (0.19)\\
\hline
0.75 & 1 & 0.813 (0.03) & 0.842 (0.05) & 0.871 (0.08) & 0.906 (0.11) & 0.94 (0.16) & 0.96 (0.18)\\
\hline
1.00 & 1 & 0.795 (0.03) & 0.819 (0.05) & 0.841 (0.08) & 0.894 (0.12) & 0.931 (0.16) & 0.955 (0.2)\\
\hline
1.25 & 1 & 0.796 (0.02) & 0.822 (0.05) & 0.818 (0.07) & 0.891 (0.11) & 0.918 (0.14) & 0.952 (0.19)\\
\hline
1.50 & 1 & 0.789 (0.03) & 0.825 (0.06) & 0.84 (0.08) & 0.87 (0.13) & 0.912 (0.15) & 0.952 (0.2)\\
\hline
0.00 & 2 & 0.818 (0.02) & 0.83 (0.07) & 0.919 (0.1) & 0.94 (0.16) & 0.958 (0.2) & 0.991 (0.23)\\
\hline
0.25 & 2 & 0.803 (0.02) & 0.815 (0.06) & 0.907 (0.1) & 0.932 (0.16) & 0.978 (0.17) & 0.988 (0.24)\\
\hline
0.50 & 2 & 0.852 (0.03) & 0.87 (0.06) & 0.885 (0.1) & 0.956 (0.13) & 0.974 (0.2) & 0.993 (0.19)\\
\hline
0.75 & 2 & 0.798 (0.03) & 0.817 (0.07) & 0.902 (0.12) & 0.916 (0.16) & 0.97 (0.18) & 0.989 (0.21)\\
\hline
1.00 & 2 & 0.79 (0.03) & 0.851 (0.07) & 0.878 (0.1) & 0.929 (0.15) & 0.955 (0.21) & 0.983 (0.24)\\
\hline
1.25 & 2 & 0.782 (0.02) & 0.858 (0.06) & 0.879 (0.12) & 0.924 (0.16) & 0.958 (0.2) & 0.981 (0.23)\\
\hline
1.50 & 2 & 0.779 (0.03) & 0.838 (0.07) & 0.888 (0.11) & 0.917 (0.16) & 0.947 (0.2) & 0.977 (0.24)\\
\hline
0.00 & 3 & 0.787 (0.02) & 0.852 (0.06) & 0.879 (0.1) & 0.936 (0.16) & 0.947 (0.2) & 0.99 (0.21)\\
\hline
0.25 & 3 & 0.782 (0.02) & 0.839 (0.06) & 0.867 (0.1) & 0.928 (0.16) & 0.97 (0.19) & 0.988 (0.22)\\
\hline
0.50 & 3 & 0.793 (0.03) & 0.856 (0.06) & 0.892 (0.12) & 0.934 (0.15) & 0.966 (0.18) & 0.988 (0.22)\\
\hline
0.75 & 3 & 0.792 (0.03) & 0.84 (0.05) & 0.875 (0.1) & 0.923 (0.15) & 0.957 (0.21) & 0.993 (0.21)\\
\hline
1.00 & 3 & 0.79 (0.04) & 0.839 (0.07) & 0.865 (0.11) & 0.916 (0.17) & 0.95 (0.21) & 0.981 (0.24)\\
\hline
1.25 & 3 & 0.791 (0.03) & 0.832 (0.07) & 0.877 (0.12) & 0.907 (0.17) & 0.953 (0.2) & 0.978 (0.24)\\
\hline
1.50 & 3 & 0.792 (0.04) & 0.827 (0.06) & 0.864 (0.12) & 0.9 (0.15) & 0.944 (0.2) & 0.975 (0.23)\\
\hline
\end{tabular}
\caption{Power and power boosts of Weighted PROCOVA over PROCOVA in the case that $\gamma_2 \neq 0$. Throughout all simulations the power of PROCOVA was $0.8$, and the power boosts of Weighted PROCOVA are in parentheses.}
\label{tab:power_table_2}
\end{sidewaystable}

\section{Case Studies}
\label{sec:case_studies}

\subsection{Description of Randomized Controlled Trials on Alzheimer's Disease}
\label{sec:case_studies_outline}

We demonstrate the advantages of Weighted PROCOVA over PROCOVA and the traditional unadjusted analysis by re-analyzing data from three RCTs on Alzheimer's disease. All three of these trials were conducted by the Alzheimer’s Disease Cooperative Study (ADCS), and were multicenter, randomized, double-blind, placebo-controlled RCTs. The ADCS is a consortium of academic medical centers and private Alzheimer’s disease clinics funded by the National Institute on Aging to conduct clinical trials on Alzheimer’s disease. The first trial investigated Docosahexaenoic Acid (DHA) supplementation, and its design and analysis were documented by \citet{quinn_et_al_2010} (ClinicalTrials.gov ID:NCT00440050). The second studied the effects of Resveratrol, and was documented by \citet{turner_2015} (ClinicalTrials.gov ID:NCT01504854). The third trial assessed the effects of Valproate and was documented by \citet{tariot_2011} (ClinicalTrials.gov ID:NCT00071721). We first summarize these RCTs, and then proceed to describe the DTG that was used in our analyses and present the results of the analyses for each study.  

The objective of the DHA trial was to assess the effects of DHA supplementation in treating cognitive and functional decline for individuals with mild to moderate Alzheimer’s disease \citep[p.~1903]{quinn_et_al_2010}. This trial was conducted from $2007$ to $2009$. In this trial, $238$ participants were randomly assigned to the active treatment arm, and the remaining $164$ participants were assigned placebo. Multiple covariates were measured at baseline, including demographics and patient characteristics (e.g., sex, age, region, and weight), lab tests (e.g., blood pressure and ApoE4 status \citep{coon_et_al_2007, safieh_et_al_2019}), and component scores of cognitive tests. The inclusion critera for this trial required subjects to be at least $50$ years old and have a baseline ``mini-mental state'' \citep[MMSE,][]{folstein_et_al_1975} between $14$ and $26$. Two co-primary endpoints were considered: changes in the Alzheimer’s Disease Assessment Scale - Cognitive Subscale \citep[ADAS-Cog 11,][]{rosen_et_al_1984} over $6$, $12$, and $18$ months, and changes in the Clinical Dementia Rating Scale Sum of Boxes scores \citep[CDR-SB,][]{morris_1992} over $6$, $12$, and $18$ months. The timepoint of interest in the original analysis was $18$ months. The original set of results for this RCT did not indicate a statistically significant improvement on the ADAS-Cog score under DHA supplementation, with the average change from baseline being $7.98$ points for the treatment group and $8.27$ points for the placebo group \citep[p.~1906--1909]{quinn_et_al_2010}.

The objective of the Resveratrol trial was to assess the effects of this treatment on plasma and CSF biomarker levels in patients with mild to moderate Alzheimer’s disease \citep[p.~1383]{turner_2015}. This trial was conducted from $2012$ to $2014$. In this trial, $64$ participants were randomly assigned to the active treatment arm, and the remaining $55$ participants were assigned placebo. Covariates measured at baseline include plasma and CSF biomarkers for Alzheimer’s disease (e.g., plasma A$\beta$40 and A$\beta$42, CSF A$\beta$40, A$\beta$42, tau, and phospho-tau 181), brain volumes (e.g., ventricular, hippocampal, and total volume), lab values (e.g., insulin, glucose, and APOE), and clinical measures (e.g., ADAS-Cog, MMSE, and CDR-SB). The inclusion critera for this trial required subjects to be at least $50$ years old and have a baseline MMSE between $14$ and $26$. The primary endpoints were adverse event counts and MRI imaging data, and ADAS-Cog and CDR-SB at the $12$ month timepoint were considered as secondary measures of disease progression. The original analyses did not demonstrate statistically significant improvements on ADAS-Cog or CDR-SB, and \citet[p.~1387]{turner_2015} noted that the RCT was underpowered for the analysis of these secondary outcomes.

The objective of the Valproate trial was to assess whether this treatment could delay or prevent the emergence of agitation and/or psychosis in participants with mild to moderate Alzheimer’s disease who had not yet experienced agitation or psychosis \citep[p.~854]{tariot_2011}. This trial was conducted from $2005$ to $2009$. In this trial, $153$ participants were randomly assigned to the active treatment arm, and the remaining $160$ participants were assigned placebo. The assignment was done in permuted blocks of four participants \citep[p.~856]{tariot_2011}. Covariates measured at baseline include demographics (e.g., age, sex, and race), vital signs (e.g., blood pressure and temperature), lab tests (e.g., ApoE4), and component scores of behavioral tests (e.g., the Neuropsychiatric Inventory \citep[NPI,][]{cummings_et_al_1994} and the Cohen-Mansfield Agitation Inventory \citep[CMAI,][]{cohen-mansfield_1986}) and cognitive tests (e.g., ADAS-Cog, MMSE, ADCS-ADL \citep{galasko_et_al_1997}, CDR-SB, and ADCS-CGIC \citep{schneider_et_al_2006}). The inclusion critera for this trial required subjects to be between $55$ and $90$ years of age, and have a baseline MMSE between $12$ and $20$. The primary endpoint was NPI at $24$ months, but secondary measures of ADAS-Cog and CDR-SB at $24$ months were also collected and analyzed. Both the ADAS-Cog and CDR-SB endpoints were measured at $6$, $12$, $18$, and $24$ months. The original analysis indicated that Valproate was statistically significantly worse than placebo for the ADAS-Cog endpoint, and that there were no differences between Valproate and placebo for the CDR-SB endpoint \citep[p.~859, 869--870]{tariot_2011}.

\subsection{Description of the Digital Twin Generator for Alzheimer's Disease}
\label{sec:case_studies_DTG}

The DTG that we employ in our three Alzheimer's case studies is Unlearn.AI's AD-DTG-3.1 \citep{unlearn_2023}. This DTG is an updated version of the new DTG architecture that was created by Unlearn.AI in June 2023 \citep{smith_2023, walsh_2023}. AD-DTG-3.1 employs a proprietary model architecture with two important features. First is that, conditional on a participant’s baseline covariates, AD-DTG-3.1 yields participant-specific expected values and variances of the control potential outcomes for the endpoints of interest at a three-month cadence post-baseline. Second is that AD-DTG-3.1 incorporates a new neural network imputation model to handle missing baseline data. This feature yields reduced sensitivity to missing baseline data, and better performance for the DTG overall. Details on the inputs, outputs, performance, and sensitivity of predictions to missing data are provided with respect to a reference dataset in the specification sheet for AD-DTG-3.1 \citep{unlearn_2023}. Additional information about the fundamental neural network architecture underlying AD-DTG-3.1 is provided by \citet{lang_et_al_2023}.  

AD-DTG-3.1 was trained on a large historical control dataset comprised of approximately 20,000 Alzheimer's patients and approximately 10,000 patients with Mild Cognitive Impairment (MCI). These data were sourced from both clinical trial control arms and observational studies. The clinical trial data used in training AD-DTG-3.1 were obtained from the Critical Path Institute’s Critical Path for Alzheimer’s Disease \citep{cpad_2023} consortium. As such, the investigators within CPAD contributed to the design and implementation of the CPAD database and/or provided data, but did not participate in the analysis of the data or the writing of this report. These data were downloaded on January 26, 2023 (which is an older version of the current release), and consist of control arms of $35$ Alzheimer’s disease clinical trials. The observational data were combined from the Alzheimer's Disease Neuroimaging Initiative (ADNI), the European Prevention of Alzheimer's Dementia Consortium (EPAD, ClinicalTrials.gov ID: NCT02804789), the National Alzheimer's Coordinating Center (NACC), and Open Access Series of Imaging Studies \citep[OASIS,][]{koenig_2020}. All of these sources include measurements of ADAS-Cog 11 and/or CDR-SB (as well as other endpoints) at three- or six-month intervals post-baseline. These training data contain the same baseline covariates as in the three case studies. Prior to examining any RCT data, AD-DTG-3.1 was fit to the change in ADAS-Cog 11 and CDR-SB at three-month intervals post-baseline (up to and including $24$ months) from the historical control data conditional on the measured covariates. After training, AD-DTG-3.1 generated prognostic scores and personalized precisions for each participant at each timepoint of interest in the three RCTs. It is important to recognize that AD-DTG-3.1 was not trained on any of the case study data, and that the sole inputs for generating prognostic scores and personalized precisions from AD-DTG-3.1 in the three RCTs are the participants' baseline covariates (not their outcomes). 

\subsection{Summary of Analyses}
\label{sec:case_studies_summary}

For each RCT, we provide the treatment effect estimates and the squares of the heteroskedasticity consistent (HC) standard errors at the timepoints of interest for both ADAS-Cog 11 and CDR-SB. We calculate HC standard errors according to the HC1 approach \citep[p.~6--7]{mackinnnon_white_1985, romano_resurrecting_2016}. The use of HC standard errors in our analyses corresponds to regulatory guidance on the use of robust standard errors as protection against model misspecification \citep{food_and_drug_administration_adjusting_2023}. All of our analyses utilized case-wise deletion in the case of missing responses. We compare inferences obtained from the prognostic covariate adjustment (abbreviated as ``PROCOVA'' in our summaries) and Weighted PROCOVA methods via the estimated percentage of variance reduction for the treatment effect estimator, and the prospective power gain for testing the treatment effect in a future trial, with respect to the unadjusted analysis. The sole predictors that are utilized in the mean model for both the PROCOVA and Weighted PROCOVA analyses are the treatment indicator and the prognostic score. Similar to the prospective power analysis of \citet[p.~1905]{quinn_et_al_2010}, in our power calculations we consider the case of a decline of $3.8$ ADAS-Cog 11 points per year in the placebo group, and set the unadjusted analysis as the baseline by identifying the parameters of the power calculation such that the unadjusted analysis will have $80\%$ power to detect a $33\%$ reduction in the rate of ADAS-Cog 11 decline in the treatment group. Although our power calculations utilize estimated variances of the treatment effects from the three RCTs, they are not post-hoc power calculations. 

The results for the DHA trial are summarized in Tables \ref{tab:DHA_mean_variance} and \ref{tab:DHA_power}. We observe from Table \ref{tab:DHA_mean_variance} that the treatment effect estimates for the PROCOVA and Weighted PROCOVA methods are similar to those of the unadjusted analysis for each endpoint at each respective timepoint. This helps to indicate that neither of these two methods introduce bias in the treatment effect estimators. We also have that Weighted PROCOVA is estimated to further reduce the variance of the treatment effect estimator from $2\%$ to $9\%$ compared to PROCOVA for the ADAS-Cog 11 endpoint, and from $0.2\%$ to $8\%$ for the CDR-SB endpoint. Table \ref{tab:DHA_power} indicates that Weighted PROCOVA can further increase the power for testing the treatment effect by $1\%$ to $3\%$ compared to PROCOVA for the ADAS-Cog 11 endpoint, and up to $3\%$ for the CDR-SB endpoint. 

\begin{table}[ht]
\centering
\begin{tabular}{| c | c | c | c | c |}
\hline
Endpoint & Analysis Method & $6$ Months & $12$ Months & $18$ Months \\
\hline
ADAS-Cog 11 & Unadjusted & $-0.415$ ($0.32$) & $-1.515$ ($0.544$) & $-0.098$ ($1.071$) \\
\hline
ADAS-Cog 11 & PROCOVA & $-0.274$ ($0.285$) & $-1.335$ ($0.496$) & $-0.079$ ($0.98$) \\
\hline
ADAS-Cog 11 & Weighted PROCOVA & $-0.423$ ($0.272$) & $-1.433$ ($0.451$) & $-0.385$ ($0.959$) \\
\hline
CDR-SB & Unadjusted & $0.119$ ($0.037$) & $-0.19$ ($0.063$) & $0.01$ ($0.111$) \\
\hline
CDR-SB & PROCOVA & $0.113$ ($0.035$) & $-0.193$ ($0.058$) & $-0.043$ ($0.101$) \\
\hline
CDR-SB & Weighted PROCOVA & $0.058$ ($0.032$) & $-0.207$ ($0.053$) & $-0.041$ ($0.101$) \\
\hline
\end{tabular}
\caption{Treatment effect estimates and the squares of their HC1 standard errors (in parentheses) for the ADAS-Cog 11 and CDR-SB endpoints at $6$, $12$, and $18$ months in the DHA trial.}
\label{tab:DHA_mean_variance}
\end{table}

\begin{table}[ht]
\centering
\begin{tabular}{| c | c | c | c | c |}
\hline
Endpoint & Analysis Method & $6$ Months & $12$ Months & $18$ Months \\
\hline
ADAS-Cog 11 & PROCOVA & $4.31\%$ & $3.45\%$ & $3.39\%$ \\
\hline
ADAS-Cog 11 & Weighted PROCOVA & $5.91\%$ & $6.75\%$ & $4.13\%$ \\
\hline
CDR-SB & PROCOVA & $2.11\%$ & $3.48\%$ & $3.51\%$ \\ 
\hline
CDR-SB & Weighted PROCOVA & $5.28\%$ & $6.43\%$ & $3.59\%$ \\
\hline
\end{tabular}
\caption{Prospective power boosts of the PROCOVA and Weighted PROCOVA methods over the unadjusted analysis for the ADAS-Cog 11 and CDR-SB endpoints at $6$, $12$, and $18$ months in the DHA trial. The unadjusted analysis was always powered to have $80\%$ power in these comparisons.}
\label{tab:DHA_power}
\end{table}

The results for the Resveratrol trial are summarized in Tables \ref{tab:resveratrol_mean_variance} and \ref{tab:resveratrol_power}. As for the DHA trial, we observe from Table \ref{tab:resveratrol_mean_variance} that the treatment effect estimates for the three methods are similar, indicating that the adjustment methods do not introduce bias. Weighted PROCOVA is estimated to further reduce the variance of the treatment effect estimator by up to $2\%$ compared to PROCOVA for the ADAS-Cog 11 endpoint, and by up to $20\%$ for the CDR-SB endpoint. Table \ref{tab:resveratrol_power} indicates that Weighted PROCOVA can further increase the power for testing the treatment effect by up to $0.5\%$ compared to PROCOVA for the ADAS-Cog 11 endpoint, and by up to $7\%$ for the CDR-SB endpoint. 

\begin{table}[ht]
\centering
\begin{tabular}{| c | c | c | c |}
\hline
Endpoint & Analysis Method & $6$ Months & $12$ Months \\
\hline
ADAS-Cog 11 & Unadjusted & $-2.356$ ($1.084$) & $-0.395$ ($1.631$) \\
\hline
ADAS-Cog 11 & PROCOVA & $-2.417$ ($1.048$) & $-0.534$ ($1.367$) \\
\hline
ADAS-Cog 11 & Weighted PROCOVA & $-2.399$ ($1.044$) & $-0.547$ ($1.345$)  \\
\hline
CDR-SB & Unadjusted & $-0.164$ ($0.091$) & $-0.208$ ($0.21$)  \\
\hline
CDR-SB & PROCOVA & $-0.16$ ($0.089$) & $-0.243$ ($0.188$)  \\
\hline
CDR-SB & Weighted PROCOVA & $-0.093$ ($0.072$) & $-0.269$ ($0.148$) \\
\hline
\end{tabular}
\caption{Treatment effect estimates and the squares of their HC1 standard errors (in parentheses) for the ADAS-Cog 11 and CDR-SB endpoints at $6$ and $12$ months in the Resveratrol trial.}
\label{tab:resveratrol_mean_variance}
\end{table}

\begin{table}[ht]
\centering
\begin{tabular}{| c | c | c | c |}
\hline
Endpoint & Analysis Method & $6$ Months & $12$ Months \\
\hline
ADAS-Cog 11 & PROCOVA & $1.32\%$ & $6.41\%$ \\
\hline
ADAS-Cog 11 & Weighted PROCOVA & $1.46\%$ & $6.96\%$ \\
\hline
CDR-SB & PROCOVA & $1.12\%$ & $4.17\%$ \\ 
\hline
CDR-SB & Weighted PROCOVA & $8.33\%$ & $11.67\%$ \\
\hline
\end{tabular}
\caption{Prospective power boosts of the PROCOVA and Weighted PROCOVA methods over the unadjusted analysis for the ADAS-Cog 11 and CDR-SB endpoints at $6$ and $12$ months in the Resveratrol trial. The unadjusted analysis was always powered to have $80\%$ power in these comparisons.}
\label{tab:resveratrol_power}
\end{table}

Finally, the results for the Valproate trial are summarized in Tables \ref{tab:valproate_mean_variance} and \ref{tab:valproate_power}. Once again, we observe that the adjustment methods do not introduce bias because the respective treatment effect estimates in Table \ref{tab:valproate_mean_variance} are similar. Weighted PROCOVA is estimated to further reduce the variance of the treatment effect estimator by up to $4\%$ compared to PROCOVA for the ADAS-Cog 11 endpoint, and by up to $3\%$ for the CDR-SB endpoint. Table \ref{tab:valproate_power} indicates that Weighted PROCOVA has similar power compared to the PROCOVA for the ADAS-Cog 11 endpoint, and improved power for the CDR-SB endpoint. For the $12$ and $24$ month timepoints, both adjustment methods have less than $80\%$ power for testing the treatment effect. 

\begin{table}[ht]
\centering
\begin{tabular}{| c | c | c | c | c | c |}
\hline
Endpoint & Analysis Method & $6$ Months & $12$ Months & $18$ months & $24$ months \\
\hline
ADAS-Cog 11 & Unadjusted & $1.189$ ($0.607$) & $2.507$ ($1.092$) & $0.528$ ($1.742$) & $1.106$ ($2.601$) \\
\hline
ADAS-Cog 11 & PROCOVA & $0.374$ ($0.574$) & $1.669$ ($1.14$) & $-0.351$ ($1.632$) & $-0.389$ ($2.649$) \\
\hline
ADAS-Cog 11 & Weighted PROCOVA &  $0.371$ ($0.577$) & $1.53$ ($1.098$) & $-0.362$ ($1.624$) & $-0.385$ ($2.642$) \\
\hline
CDR-SB & Unadjusted & $0.372$ ($0.065$) & $0.377$ ($0.137$) & $0.783$ ($0.249$) & $0.368$ ($0.267$) \\
\hline
CDR-SB & PROCOVA & $0.261$ ($0.064$) & $0.295$ ($0.134$) & $0.668$ ($0.242$) & $0.387$ ($0.251$) \\
\hline
CDR-SB & Weighted PROCOVA & $0.232$ ($0.063$) & $0.244$ ($0.13$) & 	$0.706$ ($0.235$) & $0.42$ ($0.248$) \\
\hline
\end{tabular}
\caption{Treatment effect estimates and the squares of their HC1 standard errors (in parentheses) for the ADAS-Cog 11 and CDR-SB endpoints at $6$, $12$, $18$, and $24$ months in the Valproate trial.}
\label{tab:valproate_mean_variance}
\end{table}

\begin{table}[ht]
\centering
\begin{tabular}{| c | c | c | c | c | c |}
\hline
Endpoint & Analysis Method & $6$ Months & $12$ Months & $18$ months & $24$ months \\
\hline
ADAS-Cog 11 & PROCOVA & $2.15\%$ & $-1.69\%$ & $2.48\%$ & $-0.72\%$ \\
\hline
ADAS-Cog 11 & Weighted PROCOVA & $1.99\%$ & $-0.21\%$ & $2.66\%$ & $-0.62\%$ \\
\hline
CDR-SB & PROCOVA & $0.67\%$ & $0.96\%$ & $0.98\%$ & $2.39\%$ \\
\hline
CDR-SB & Weighted PROCOVA & $1.01\%$ & $2.27\%$ & $2.08\%$ & $2.73\%$ \\
\hline
\end{tabular}
\caption{Prospective power boosts of the PROCOVA and Weighted PROCOVA methods over the unadjusted analysis for the ADAS-Cog 11 and CDR-SB endpoints at $6$, $12$, and $18$ months in the Valproate trial. The unadjusted analysis was always powered to have $80\%$ power in these comparisons.}
\label{tab:valproate_power}
\end{table}

\section{Concluding Remarks}
\label{sec:concluding_remarks}

The limitations, challenges, and long timelines associated with RCTs can be effectively resolved via generative AI-powered digital twins and statistical methods that utilize the rich set of information uniquely provided by DTGs. Training a DTG on a historical control dataset with a large number of covariates enables one to leverage a small set of summaries from the digital twin distributions as powerful adjustment factors to explain substantial variation in the outcomes. In the case of heteroskedasticity, one can adjust for both the expectation and the variance of the digital twin distributions via Weighted PROCOVA to increase the efficiency and power of the treatment effect inferences. As we have demonstrated theoretically and empirically, our method can yield point estimators for treatment effects that are unbiased and have less variance, and hypothesis tests for treatment effects with controlled Type I error rates that are more powerful, compared to PROCOVA or unadjusted methods that do not account for the variances of the digital twin distributions. Weighted PROCOVA had three additional attractive features. First, the DTG is trained on external control data, so that it effectively incorporates additional information beyond that in the RCT. Second, it does not introduce bias in the analysis because the external control data are independent of the RCT and the inputs for the DTG in the RCT are the participants' baseline covariates. Third, the logic underlying the method is intuitive and easy to explain, as it assigns more weight to participants for whom the DTG is more confident in terms of predicting their control potential outcomes, and less weight to participants for whom the DTG is less confident. Our analyses of three RCTs further demonstrate the consistency in the benefits and features of Weighted PROCOVA.

The combination of generative AI and Weighted PROCOVA can power the next generation of RCTs and accelerate drug development. They can be seamlessly integrated into both the design and analysis of a RCT, maintaining randomization and its benefits while contributing additional powerful adjustment features to reduce the uncertainty and boost the power of treatment effect inferences. In particular, digital twin distributions can be prospectively computed for RCT participants prior to treatment assignments. Furthermore, as Weighted PROCOVA yields valid inferences that meet regulatory guidance on bias, Type I error rates, and confidence interval coverage, adjustment for both the prognostic score and personalized precision can be incorporated into the primary or secondary analyses of the RCT. The benefits of this combination can also be estimated prospectively via external control data. Ultimately, combining generative AI with Weighted PROCOVA can enable trial designers to reduce the number of participants, and shorten the enrollment timeline, for a RCT without sacrificing the validity of the analysis or decreasing the likelihood of producing evidence suitable for making regulatory decisions.

\section*{Data Availability}

The data used in this study are available from the following sources, subject to their discretion.

Certain data used in the preparation of this article were obtained from the Alzheimer’s Disease Neuroimaging Initiative (ADNI) database (adni.loni.usc.edu). The ADNI was launched in 2003 as a public-private partnership, led by Principal Investigator Michael W. Weiner, MD. The primary goal of ADNI has been to test whether serial magnetic resonance imaging (MRI), positron emission tomography (PET), other biological markers, and clinical and neuropsychological assessment can be combined to measure the progression of mild cognitive impairment (MCI) and early Alzheimer’s disease (AD). A complete listing of ADNI investigators can be found at: http://adni.loni.usc.edu/wp-content/uploads/how\_to\_apply/ADNI\_Acknowledgement\_List.pdf.

Certain data used in the preparation of this article for the Quinn et al, 2010, Tariot et al, 2011, and Turner et al, 2015 studies were obtained from the University of California, San Diego Alzheimer’s Disease Cooperative Study Legacy database.

Certain data used in preparation of this article were obtained from the EPAD LCS data set V.IMI, doi:10.34688/epadlcs\_v.imi\_20.10.30. The EPAD LCS was launched in 2015 as a public-private partnership, led by Chief Investigator Professor Craig Ritchie MB BS. The primary research goal of the EPAD LCS is to provide a well-phenotyped probability-spectrum population for developing and continuously improving disease models for Alzheimer’s disease in individuals without dementia. A complete list of EPAD Investigators can be found at: http://ep-ad.org/wp-content/uploads/2020/12/202010\_List-of-epadistas.pdf.

\section*{Acknowledgments}

Data collection and sharing for this project was funded in part by the Alzheimer’s Disease Neuroimaging Initiative (ADNI) (National Institutes of Health Grant U01 AG024904) and DOD ADNI (Department of Defense award number W81XWH-12-2-0012). ADNI is funded by the National Institute on Aging, the National Institute of Biomedical Imaging and Bioengineering, and through generous contributions from the following: AbbVie, Alzheimer’s Association; Alzheimer’s Drug Discovery Foundation; Araclon Biotech; BioClinica, Inc.; Biogen; Bristol-Myers Squibb Company; CereSpir, Inc.; Cogstate; Eisai Inc.; Elan Pharmaceuticals, Inc.; Eli Lilly and Company; EuroImmun; F. Hoffmann-La Roche Ltd and its affiliated company Genentech, Inc.; Fujirebio; GE Healthcare; IXICO Ltd.; Janssen Alzheimer Immunotherapy Research \& Development, LLC.; Johnson \& Johnson Pharmaceutical Research \& Development LLC.; Lumosity; Lundbeck; Merck \& Co., Inc.; Meso Scale Diagnostics, LLC.; NeuroRx Research; Neurotrack Technologies; Novartis Pharmaceuticals Corporation; Pfizer Inc.; Piramal Imaging; Servier; Takeda Pharmaceutical Company; and Transition Therapeutics. The Canadian Institutes of Health Research is providing funds to support ADNI clinical sites in Canada. Private sector contributions are facilitated by the Foundation for the National Institutes of Health (www.fnih.org). The grantee organization is the Northern California Institute for Research and Education, and the study is coordinated by the Alzheimer’s Therapeutic Research Institute at the University of Southern California. ADNI data are disseminated by the Laboratory for Neuro Imaging at the University of Southern California. 

Data collection and sharing for this project was funded in part by the University of California, San Diego Alzheimer’s Disease Cooperative Study (ADCS) (National Institute on Aging Grant Number U19AG010483). 

This work used data and/or samples from the EPAD project which received support from the EU/EFPIA Innovative Medicines Initiative Joint Undertaking EPAD grant agreement no 115736 and an Alzheimer’s Association Grant (SG-21-818099-EPAD). 

Certain data used in the preparation of this article were provided [in part] by OASIS [OASIS-4: Clinical Cohort: Principal Investigators: T. Benzinger, L. Koenig, P. LaMontagne].

Certain data used in the preparation of this article were obtained from the National Alzheimer's Coordinating Center (NACC) database. The NACC database is funded by NIA/NIH Grant U24 AG072122. NACC data are contributed by the NIA-funded ADRCs: P30 AG062429 (PI James Brewer, MD, PhD), P30 AG066468 (PI Oscar Lopez, MD), P30 AG062421 (PI Bradley Hyman, MD, PhD), P30 AG066509 (PI Thomas Grabowski, MD), P30 AG066514 (PI Mary Sano, PhD), P30 AG066530 (PI Helena Chui, MD), P30 AG066507 (PI Marilyn Albert, PhD), P30 AG066444 (PI John Morris, MD), P30 AG066518 (PI Jeffrey Kaye, MD), P30 AG066512 (PI Thomas Wisniewski, MD), P30 AG066462 (PI Scott Small, MD), P30 AG072979 (PI David Wolk, MD), P30 AG072972 (PI Charles DeCarli, MD), P30 AG072976 (PI Andrew Saykin, PsyD), P30 AG072975 (PI David Bennett, MD), P30 AG072978 (PI Neil Kowall, MD), P30 AG072977 (PI Robert Vassar, PhD), P30 AG066519 (PI Frank LaFerla, PhD), P30 AG062677 (PI Ronald Petersen, MD, PhD), P30 AG079280 (PI Eric Reiman, MD), P30 AG062422 (PI Gil Rabinovici, MD), P30 AG066511 (PI Allan Levey, MD, PhD), P30 AG072946 (PI Linda Van Eldik, PhD), P30 AG062715 (PI Sanjay Asthana, MD, FRCP), P30 AG072973 (PI Russell Swerdlow, MD), P30 AG066506 (PI Todd Golde, MD, PhD), P30 AG066508 (PI Stephen Strittmatter, MD, PhD), P30 AG066515 (PI Victor Henderson, MD, MS), P30 AG072947 (PI Suzanne Craft, PhD), P30 AG072931 (PI Henry Paulson, MD, PhD), P30 AG066546 (PI Sudha Seshadri, MD), P20 AG068024 (PI Erik Roberson, MD, PhD), P20 AG068053 (PI Justin Miller, PhD), P20 AG068077 (PI Gary Rosenberg, MD), P20 AG068082 (PI Angela Jefferson, PhD), P30 AG072958 (PI Heather Whitson, MD), P30 AG072959 (PI James Leverenz, MD).

\section*{Financial Disclosure}
AMV, AAV, JLR, AS, JRW, and CKF are equity-holding employees of Unlearn.AI, Inc., a company that creates software for clinical research.

\appendix

\section{Derivations and Proofs}
\label{sec:appendix}

\subsection{Skedastic Function Model Specification and Coefficient Estimators}
\label{sec:appendix_skedastic_function_model_and_estimators}

In all of our derivations and proofs, the skedastic function model is specified as
\begin{equation}
\label{eq:appendix_skedastic_function_model}
\mathrm{log} \left ( e_i^2 \right ) = \gamma_0 + \gamma_1 \mathrm{log} \left ( s_i^2 \right ) + \xi_i,
\end{equation}
where $s_i^2$ is the variance of the probability distribution from the DTG for participant $i$. We define $u_i = \left ( 1, \mathrm{log} \left ( s_i^2 \right ) \right )^{\mathsf{T}}$ and $\mathbf{U}$ as the $N \times 2$ matrix whose $i$th row is $u_i^{\mathsf{T}}$. We let $d_i \in \mathbb{R}^N$ be the vector that has zero in all of its entries except for the $i$th entry, which is equal to $1$. The estimators for the coefficients in equation (\ref{eq:appendix_skedastic_function_model}) that we consider throughout are obtained via $\hat{\gamma} = \left ( \mathbf{U}^{\mathsf{T}} \mathbf{U} \right )^{-1} \mathbf{U}^{\mathsf{T}} \mathrm{log} \left \{ \left ( y - \mathbf{H} y \right )^2 \right \}$, where the logarithmic operation in this expression is performed entrywise. Specifically, 
\[
\begin{pmatrix} \hat{\gamma}_0 \\ \hat{\gamma}_1 \end{pmatrix} = \left ( \mathbf{U}^{\mathsf{T}}\mathbf{U} \right )^{-1}\mathbf{U}^{\mathsf{T}} \begin{pmatrix} \mathrm{log} \left [ \left \{  d_1^{\mathsf{T}} \left ( I - \mathbf{V} \left ( \mathbf{V}^{\mathsf{T}}\mathbf{V} \right )^{-1}\mathbf{V}^{\mathsf{T}} \right ) y \right \}^2 \right ] \\ \vdots \\ \mathrm{log} \left [ \left \{  d_N^{\mathsf{T}} \left ( I - \mathbf{V} \left ( \mathbf{V}^{\mathsf{T}}\mathbf{V} \right )^{-1}\mathbf{V}^{\mathsf{T}} \right ) y \right \}^2 \right ] \end{pmatrix}.
\]

\subsection{The Distributions of Residuals and Their Transformations}
\label{sec:appendix_residuals_proofs}

As $\left [ y \mid \mathbf{U}, \mathbf{V} \right ]$ is Multivariate Normal, the distribution of the vector of residuals conditional on the predictors, i.e., $\left [ \left \{ I - \mathbf{V} \left ( \mathbf{V}^{\mathsf{T}}\mathbf{V} \right )^{-1}\mathbf{V}^{\mathsf{T}} \right \}y \mid \mathbf{U}, \mathbf{V} \right ]$, is Multivariate Normal. We need only calculate $\mathbb{E} \left [ \left \{ I - \mathbf{V} \left ( \mathbf{V}^{\mathsf{T}}\mathbf{V} \right )^{-1}\mathbf{V}^{\mathsf{T}} \right \}y \mid \mathbf{U}, \mathbf{V} \right ]$ and $\mathrm{Cov} \left [ \left \{ I - \mathbf{V} \left ( \mathbf{V}^{\mathsf{T}}\mathbf{V} \right )^{-1}\mathbf{V}^{\mathsf{T}} \right \}y \mid \mathbf{U}, \mathbf{V} \right ]$ to fully derive the distribution of the residuals. As $I - \mathbf{V} \left ( \mathbf{V}^{\mathsf{T}}\mathbf{V} \right )^{-1}\mathbf{V}^{\mathsf{T}}$ is a projection matrix into the orthogonal complement of the column space of $\mathbf{V}$, we have $\mathbb{E} \left [ \left \{ I - \mathbf{V} \left ( \mathbf{V}^{\mathsf{T}}\mathbf{V} \right )^{-1}\mathbf{V}^{\mathsf{T}} \right \}y \mid \mathbf{U}, \mathbf{V} \right ] = \left \{ I - \mathbf{V} \left ( \mathbf{V}^{\mathsf{T}}\mathbf{V} \right )^{-1}\mathbf{V}^{\mathsf{T}} \right \}\mathbf{V}\beta = \left (0, \ldots, 0 \right )^{\mathsf{T}}$. For the second calculation,
\begin{equation}
\label{eq:residuals_covariance_matrix}
\mathrm{Cov} \left [ \left \{ I - \mathbf{V} \left ( \mathbf{V}^{\mathsf{T}}\mathbf{V} \right )^{-1}\mathbf{V}^{\mathsf{T}} \right \}y \mid \mathbf{U}, \mathbf{V} \right ] = \left \{ I - \mathbf{V} \left ( \mathbf{V}^{\mathsf{T}}\mathbf{V} \right )^{-1}\mathbf{V}^{\mathsf{T}} \right \} \Omega \left \{ I - \mathbf{V} \left ( \mathbf{V}^{\mathsf{T}}\mathbf{V} \right )^{-1}\mathbf{V}^{\mathsf{T}} \right \},
\end{equation}
where $\Omega$ is the $N \times N$ diagonal matrix whose $i$th diagonal entry is $\sigma_i^2$. The $(i,i)$ entry of the matrix in equation (\ref{eq:residuals_covariance_matrix}) is $d_i^{\mathsf{T}} \left \{ I - \mathbf{V} \left ( \mathbf{V}^{\mathsf{T}} \mathbf{V} \right )^{-1}\mathbf{V}^{\mathsf{T}} \right \} \Omega \left \{ I - \mathbf{V} \left ( \mathbf{V}^{\mathsf{T}} \mathbf{V} \right )^{-1}\mathbf{V}^{\mathsf{T}} \right \} d_i$, and is directly calculated as $\left ( 1 - h_{ii} \right )^2 \sigma_i^2 + \displaystyle \sum_{k \neq i} h_{ik}^2\sigma_k^2$. Similarly, the $(i,j)$ entry is directly calculated as $\displaystyle \sum_{k \neq i, j} h_{ik}h_{jk}\sigma_k^2 - h_{ij} \left ( 1 - h_{ii} \right )\sigma_i^2 - h_{ij} \left ( 1 - h_{jj} \right ) \sigma_j^2$ for $i \neq j$. Thus, for distinct $i, j \in \{1, \ldots, N\}$, $\mathrm{Var} \left ( e_i \mid \mathbf{U}, \mathbf{V} \right ) = \left ( 1 - h_{ii} \right )^2 \sigma_i^2 + \displaystyle \sum_{k \neq i} h_{ik}^2 \sigma_k^2$ and $\mathrm{Cov} \left ( e_i, e_j \mid \mathbf{U}, \mathbf{V} \right ) = \displaystyle \sum_{k \neq i, j} h_{ik}h_{jk}\sigma_k^2 - h_{ij} \left ( 1 - h_{ii} \right )\sigma_i^2 - h_{ij} \left ( 1 - h_{jj} \right ) \sigma_j^2$.

We represent the joint distribution of the squared residuals conditional on $\mathbf{U}$ and $\mathbf{V}$ via
\begin{equation}
\label{eq:squared_residual_representation}
e_i^2 = \left \{ \left ( 1 - h_{ii} \right )^2 \sigma_i^2 + \sum_{k \neq i} h_{ik}^2 \sigma_k^2 \right \} Z_i^2,
\end{equation}
where $\left [ \left ( Z_1, \ldots, Z_N \right )^{\mathsf{T}} \mid \mathbf{U}, \mathbf{V} \right ]$ is Multivariate Normal with $\mathbb{E} \left \{ \left ( Z_1, \ldots, Z_N \right )^{\mathsf{T}} \mid \mathbf{U}, \mathbf{V} \right \} = (0, \ldots, 0)^{\mathsf{T}}$, $\mathrm{Var} \left ( Z_i \mid \mathbf{U}, \mathbf{V} \right ) = 1$, and
\[
\mathrm{Cov} \left ( Z_i, Z_j \mid \mathbf{U}, \mathbf{V} \right ) = \frac{\displaystyle \sum_{k \neq i, j} h_{ik}h_{jk}\sigma_k^2 - h_{ij} \left ( 1 - h_{ii} \right )\sigma_i^2 - h_{ij} \left ( 1 - h_{jj} \right ) \sigma_j^2 }{\sqrt{\left ( 1 - h_{ii} \right )^2 \sigma_i^2 + \displaystyle \sum_{k \neq i} h_{ik}^2 \sigma_k^2} \sqrt{\left ( 1 - h_{jj} \right )^2 \sigma_j^2 + \displaystyle \sum_{k \neq j} h_{jk}^2 \sigma_k^2}}.
\]
We calculate $\mathrm{Cov} \left ( e_i^2, e_j^2 \mid \mathbf{U}, \mathbf{V} \right )$ according to

\begin{align*}
\mathrm{Cov} \left ( e_i^2, e_j^2 \mid \mathbf{U}, \mathbf{V} \right ) &= \mathbb{E} \left ( e_i^2 e_j^2 \mid \mathbf{U}, \mathbf{V} \right ) - \left \{ \left ( 1 - h_{ii} \right )^2 \sigma_i^2 + \displaystyle \sum_{k \neq i}  h_{ik}^2 \sigma_k^2 \right \} \left \{ \left ( 1 - h_{jj} \right )^2 \sigma_j^2 + \displaystyle \sum_{k \neq j} h_{jk}^2 \sigma_k^2 \right \} \\
&= \mathbb{E} \left \{ e_i^2 \mathbb{E} \left( e_j^2 \mid e_i, \mathbf{U}, \mathbf{V} \right ) \mid \mathbf{U}, \mathbf{V} \right \} \\
& \ \ \ - \left \{ \left ( 1 - h_{ii} \right )^2 \sigma_i^2 + \displaystyle \sum_{k \neq i} h_{ik}^2 \sigma_k^2 \right \} \left \{ \left ( 1 - h_{jj} \right )^2 \sigma_j^2 + \displaystyle \sum_{k \neq j} h_{jk}^2 \sigma_k^2 \right \} \\
&= \mathbb{E} \Bigg{\{} e_i^2 \Bigg{(} \left ( 1 - h_{jj} \right )^2 \sigma_j^2 + \sum_{k \neq j} h_{jk}^2 \sigma_k^2 \\
& \ \ \ \ \ \ \ - \frac{ \left \{ \displaystyle \sum_{k \neq i, j} h_{ik}h_{jk}\sigma_k^2 - h_{ij} \left ( 1 - h_{ii} \right ) \sigma_i^2 - h_{ij} \left ( 1 - h_{jj} \right ) \sigma_j^2 \right \}^2}{\left ( 1 - h_{ii} \right )^2 \sigma_i^2 + \displaystyle \sum_{k \neq i} h_{ik}^2 \sigma_k^2} \\
& \ \ \ \ \ \ \ + \frac{ \left \{ \displaystyle \sum_{k \neq i, j} h_{ik}h_{jk}\sigma_k^2 - h_{ij} \left ( 1 - h_{ii} \right ) \sigma_i^2 - h_{ij} \left ( 1 - h_{jj} \right ) \sigma_j^2 \right \}^2}{\left \{ \left ( 1 - h_{ii} \right )^2 \sigma_i^2 + \displaystyle \sum_{k \neq i} h_{ik}^2 \sigma_k^2 \right \}^2}e_i^2 \Bigg{)} \mid \mathbf{U}, \mathbf{V} \Bigg{\}} \\
& \ \ \ - \left \{ \left ( 1 - h_{ii} \right )^2 \sigma_i^2 + \displaystyle \sum_{k \neq i} h_{ik}^2 \sigma_k^2 \right \} \left \{  \left ( 1 - h_{jj} \right )^2 \sigma_j^2 + \displaystyle \sum_{k \neq j} h_{jk}^2 \sigma_k^2 \right \} \\
&= 2 \left \{ \displaystyle \sum_{k \neq i, j} h_{ik}h_{jk}\sigma_k^2 - h_{ij} \left ( 1 - h_{ii} \right )\sigma_i^2 - h_{ij} \left ( 1 - h_{jj} \right ) \sigma_j^2 \right \}^2.
\end{align*}

Equation (\ref{eq:squared_residual_representation}) enables us to represent the joint distribution of the logarithmic transformation of the squared residuals conditional on $\mathbf{U}$ and $\mathbf{V}$ via
\begin{equation}
\label{eq:log_squared_residual_representation}
\mathrm{log} \left ( e_i^2 \right ) = \mathrm{log} \left ( \left ( 1 - h_{ii} \right )^2 \sigma_i^2 + \sum_{k \neq i} h_{ik}^2 \sigma_k^2 \right ) + \mathrm{log} \left ( Z_i^2 \right ).
\end{equation}
This representation yields
\begin{align*}
\mathbb{E} \left \{ \mathrm{log} \left ( e_i^2 \right ) \mid \mathbf{U}, \mathbf{V} \right \} &= \mathrm{log} \left ( \left ( 1 - h_{ii} \right )^2 \sigma_i^2  + \sum_{k \neq i} h_{ik}^2 \sigma_k^2 \right ) + \mathbb{E} \left \{ \mathrm{log} \left ( Z_i^2 \right ) \mid \mathbf{U}, \mathbf{V} \right \} \\
&= \mathrm{log} \left ( \left ( 1 - h_{ii} \right )^2 \sigma_i^2 + \sum_{k \neq i} h_{ik}^2 \sigma_k^2 \right ) - (\gamma_{\mathrm{EM}} + \mathrm{log} 2 ),
\end{align*}
\noindent where $\gamma_{\mathrm{EM}} \approx 0.577$ denotes the Euler-Mascheroni constant (and does not correspond to a skedastic function model coefficient). In addition, $\mathrm{Var} \left \{ \mathrm{log} \left ( e_i^2 \right ) \mid \mathbf{U}, \mathbf{V} \right \} = \mathrm{Var} \left \{ \mathrm{log} \left ( Z_i^2 \right ) \mid \mathbf{U}, \mathbf{V} \right \} = \pi^2/2$. From the Cauchy-Schwarz inequality, $-\pi^2/2 - \left ( \gamma_{\mathrm{EM}} + \mathrm{log}2 \right )^2 \leq \mathbb{E} \left \{ \mathrm{log} \left ( Z_i^2 \right ) \mathrm{log} \left ( Z_j^2 \right ) \mid \mathbf{U}, \mathbf{V} \right \} \leq \pi^2/2 + \left ( \gamma_{\mathrm{EM}} + \mathrm{log} 2 \right )^2$. Similarly, as $\mathbb{E} \left \{ \mathrm{log} \left ( e_i^2 \right ) \mid \mathbf{U}, \mathbf{V} \right \}$ and $\mathrm{Var} \left \{ \mathrm{log} \left ( e_i^2 \right ) \mid \mathbf{U}, \mathbf{V} \right \}$ are both finite, $\bigg{|} \mathbb{E} \left \{ \mathrm{log} \left ( e_i^2 \right ) \mathrm{log} \left ( e_j^2 \right ) \mid \mathbf{U}, \mathbf{V} \right \} \bigg{|} < \sqrt{\mathbb{E} \left [ \left \{ \mathrm{log} \left ( e_i^2 \right ) \right \}^2 \mid \mathbf{U}, \mathbf{V} \right ] \mathbb{E} \left [ \left \{ \mathrm{log} \left ( e_j^2 \right ) \right \}^2 \mid \mathbf{U}, \mathbf{V} \right ]}$, and so the $\mathrm{log} \left ( e_i^2 \right )$ and $\mathrm{log} \left ( e_j^2 \right )$ have finite covariance conditional on $\mathbf{U}, \mathbf{V}$.

\subsection{Proofs of the Frequentist Properties of the Coefficient Estimators}
\label{sec:appendix_frequentist_proofs}

The frequentist properties of the coefficient estimators are derived via the representation of the logarithmic transformation of the squared residuals in equation (\ref{eq:log_squared_residual_representation}). Specifically, we represent the coefficient estimators as
\begin{equation}
\label{eq:coefficient_estimators_representation}
\begin{pmatrix} \hat{\gamma}_0 \\ \hat{\gamma}_1 \end{pmatrix} = \left ( \mathbf{U}^{\mathsf{T}}\mathbf{U} \right )^{-1}\mathbf{U}^{\mathsf{T}} \begin{pmatrix} \mathrm{log} \left ( \left ( 1 - h_{11} \right )^2 \sigma_1^2 + \displaystyle \sum_{k \neq 1} h_{1k}^2 \sigma_k^2 \right ) \\ \vdots \\ \mathrm{log} \left ( \left ( 1 - h_{NN} \right )^2 \sigma_N^2 + \displaystyle \sum_{k \neq N} h_{Nk}^2 \sigma_k^2 \right ) \end{pmatrix} + \left ( \mathbf{U}^{\mathsf{T}}\mathbf{U} \right )^{-1}\mathbf{U}^{\mathsf{T}} \begin{pmatrix} \mathrm{log} \left ( Z_1^2 \right ) \\ \vdots \\ \mathrm{log} \left ( Z_N^2 \right ) \end{pmatrix}.
\end{equation}
\noindent Letting $\overline{\mathrm{log} \left ( s^2 \right )} = \displaystyle \frac{1}{N} \sum_{k=1}^N \mathrm{log} \left ( s_k^2 \right )$ and $c =  \left [ N  \displaystyle \sum_{k=1}^N \left \{ \mathrm{log} \left ( s_k^2 \right ) - \overline{\mathrm{log} \left ( s^2 \right )} \right \}^2 \right ]^{-1}$, we have 
\[
\left ( \mathbf{U}^{\mathsf{T}}\mathbf{U} \right )^{-1} = c \begin{pmatrix} \displaystyle \sum_{k=1}^N \left \{ \mathrm{log} \left ( s_k^2 \right ) \right \}^2 & \displaystyle -N\overline{\mathrm{log} \left ( s^2 \right )} \\ \displaystyle -N\overline{\mathrm{log} \left ( s^2 \right )} & N \end{pmatrix}
\]
so that
\begin{align*}
\left ( \mathbf{U}^{\mathsf{T}}\mathbf{U} \right )^{-1}\mathbf{U}^{\mathsf{T}} &= c\begin{pmatrix} \displaystyle \left \{ \sum_{k=1}^N \left \{ \mathrm{log} \left ( s_k^2 \right ) \right \}^2 - N\mathrm{log} \left ( s_1^2 \right ) \overline{\mathrm{log} \left ( s^2 \right )} \right \} & \cdots & \left \{ \displaystyle \sum_{k=1}^N \left \{ \mathrm{log} \left ( s_k^2 \right ) \right \}^2 - N \mathrm{log} \left ( s_N^2 \right ) \overline{\mathrm{log} \left ( s^2 \right )} \right \} \\ \left \{ N\mathrm{log} \left ( s_1^2 \right ) - N\overline{\mathrm{log} \left ( s^2 \right )} \right \} & \cdots & \left \{ N\mathrm{log} \left ( s_N^2 \right ) - N\overline{\mathrm{log} \left ( s^2 \right )} \right \} \end{pmatrix}.
\end{align*}

\begin{proof}[Proof of Proposition \ref{lem:conditional_expectations}]
We first calculate  
\begin{equation}
\label{eq:coefficient_estimators_representation_1}
\left ( \mathbf{U}^{\mathsf{T}}\mathbf{U} \right )^{-1}\mathbf{U}^{\mathsf{T}} \begin{pmatrix} \mathrm{log} \left \{ \left ( 1 - h_{11} \right )^2 \sigma_1^2 + \displaystyle \sum_{k \neq 1} h_{1k}^2\sigma_k^2 \right \} \\ \vdots \\ \mathrm{log} \left \{ \left ( 1 - h_{NN} \right )^2 \sigma_N^2 + \displaystyle \sum_{k \neq N}  h_{Nk}^2 \sigma_k^2 \right \} \end{pmatrix} 
\end{equation}
\noindent and
\begin{equation}
\label{eq:coefficient_estimators_representation_2}
\mathbb{E} \left \{ \left ( \mathbf{U}^{\mathsf{T}}\mathbf{U} \right )^{-1}\mathbf{U}^{\mathsf{T}} \begin{pmatrix} \mathrm{log} \left ( Z_1^2 \right ) \\ \vdots \\ \mathrm{log} \left ( Z_N^2 \right )^2 \end{pmatrix} \mid \mathbf{U}, \mathbf{V} \right \}.
\end{equation}
The two entries in equation (\ref{eq:coefficient_estimators_representation_1}) are 
\begin{align*}
\frac{1}{N \displaystyle \sum_{k=1}^N \left \{ \mathrm{log} \left ( s_k^2 \right ) - \overline{\mathrm{log} \left ( s^2 \right )} \right \}^2} & \Bigg{(} \left [ \displaystyle \sum_{k=1}^N \left \{ \mathrm{log} \left ( s_k^2 \right ) \right \}^2 \right ] \displaystyle \sum_{i=1}^N \mathrm{log} \left \{ \left ( 1 - h_{ii} \right )^2 \sigma_i^2 + \displaystyle \sum_{k \neq i} h_{ik}^2 \sigma_k^2 \right \} \\
& \ \ \ - N \overline{\mathrm{log} \left ( s^2 \right )} \displaystyle \sum_{i=1}^N\mathrm{log} \left ( s_i^2 \right ) \mathrm{log} \left \{ \left ( 1 - h_{ii} \right )^2 \sigma_i^2 + \displaystyle \sum_{k \neq i} h_{ik}^2 \sigma_k^2 \right \} \Bigg{)}
\end{align*}
and 
\begin{align*}
\frac{1}{N \displaystyle \sum_{k=1}^N \left \{ \mathrm{log} \left ( s_k^2 \right ) - \overline{\mathrm{log} \left ( s^2 \right )} \right \}^2} & \Bigg{(} N \displaystyle \sum_{i=1}^N \mathrm{log} \left ( s_i^2 \right ) \mathrm{log} \left \{ \left ( 1 - h_{ii} \right )^2 \sigma_i^2 + \sum_{k \neq i} h_{ik}^2 \sigma_k^2 \right \} \\
& \ \ \ - N \overline{\mathrm{log} \left ( s^2 \right )} \sum_{i=1}^N \mathrm{log} \left \{ \left ( 1 - h_{ii} \right )^2 \sigma_i^2 + \sum_{k \neq i} h_{ik}^2 \sigma_k^2 \right \}  \Bigg{)}.
\end{align*}
\noindent In addition, the two entries in $\left ( \mathbf{U}^{\mathsf{T}}\mathbf{U} \right )^{-1}\mathbf{U}^{\mathsf{T}} \begin{pmatrix} \mathrm{log} \left ( Z_1^2 \right ) \\ \vdots \\ \mathrm{log} \left ( Z_N^2 \right ) \end{pmatrix}$ are
\[
\frac{1}{N \displaystyle \sum_{k=1}^N \left \{ \mathrm{log} \left ( s_k^2 \right ) - \overline{\mathrm{log} \left ( s^2 \right )} \right \}^2} \left ( \left [ \displaystyle \sum_{k=1}^N \left \{ \mathrm{log} \left ( s_k^2 \right ) \right \}^2 \right ] \displaystyle \sum_{k=1}^N \mathrm{log} \left ( Z_k^2 \right ) - N \overline{\mathrm{log} \left ( s^2 \right )} \displaystyle \sum_{k=1}^N \mathrm{log} \left ( s_k^2 \right ) \mathrm{log} \left ( Z_k^2 \right ) \right ) 
\]
\noindent and
\[
\frac{1}{N \displaystyle \sum_{k=1}^N \left \{ \mathrm{log} \left ( s_k^2 \right ) - \overline{\mathrm{log} \left ( s^2 \right )} \right \}^2} \left \{ N \displaystyle \sum_{k=1}^N \mathrm{log} \left ( s_k^2 \right )  \mathrm{log} \left ( Z_k^2 \right ) - N\overline{\mathrm{log} \left ( s^2 \right )} \displaystyle \sum_{k=1}^N \mathrm{log} \left ( Z_k^2 \right ) \right \}.  
\]
\noindent Thus
\begin{align*}
\mathbb{E} \left ( \hat{\gamma}_0 \mid \mathbf{U}, \mathbf{V} \right ) &= 
\frac{1}{N \displaystyle \sum_{k=1}^N \left \{ \mathrm{log} \left ( s_k^2 \right ) - \overline{\mathrm{log} \left ( s^2 \right )} \right \}^2} \Bigg{(} \left [ \displaystyle \sum_{k=1}^N \left \{ \mathrm{log} \left ( s_k^2 \right ) \right \}^2 \right ] \displaystyle \sum_{i=1}^N \mathrm{log} \left \{ \left ( 1 - h_{ii} \right )^2 \sigma_i^2 + \displaystyle \sum_{k \neq i} h_{ik}^2 \sigma_k^2 \right \} \\
& \ \ \ \ \ \ \ \ \ \ \ \ \ \ \ \ \ \ \ \ \ \ \ \ \ \ \ \ \ \ \ \ \ \ \ \ \ \ \ \ \ \ \ - N\overline{\mathrm{log} \left ( s^2 \right )}  \displaystyle \sum_{i=1}^N \mathrm{log} \left ( s_i^2 \right ) \mathrm{log} \left \{ \left ( 1 - h_{ii} \right )^2 \sigma_i^2 + \displaystyle \sum_{k \neq i} h_{ik}^2 \sigma_k^2 \right \} \Bigg{)} \\
& \ \ \ + \frac{1}{N \displaystyle \sum_{k=1}^N \left \{ \mathrm{log} \left ( s_k^2 \right ) - \overline{\mathrm{log} \left ( s^2 \right )} \right \}^2}  \Bigg{(}\ \left [\displaystyle \sum_{k=1}^N \left \{ \mathrm{log} \left ( s_k^2 \right ) \right \}^2 \right ] \displaystyle \sum_{k=1}^N \mathbb{E} \left \{ \mathrm{log} \left ( Z_k^2 \right ) \mid \mathbf{U}, \mathbf{V} \right \} \\
& \ \ \ \ \ \ \ \ \ \ \ \ \ \ \ \ \ \ \ \ \ \ \ \ \ \ \ \ \ \ \ \ \ \ \ \ \ \ \ \ \ \ \ \ \ \ \ - N \overline{\mathrm{log} \left ( s^2 \right ) } \displaystyle \sum_{k=1}^N \mathrm{log} \left ( s_k^2 \right ) \mathbb{E} \left \{ \mathrm{log} \left ( Z_k^2 \right ) \mid \mathbf{U}, \mathbf{V} \right \}\Bigg{)}
\end{align*}
and
\begin{align*}
\mathbb{E} \left ( \hat{\gamma}_1 \mid \mathbf{U}, \mathbf{V} \right ) &= \frac{1}{N \displaystyle \sum_{k=1}^N \left \{ \mathrm{log} \left ( s_k^2 \right ) - \overline{\mathrm{log} \left ( s^2 \right )} \right \}^2} \Bigg{(} N \displaystyle \sum_{i=1}^N \mathrm{log} \left ( s_i^2 \right ) \mathrm{log} \left \{ \left ( 1 - h_{ii} \right )^2 \sigma_i^2 + \sum_{k \neq i} h_{ik}^2 \sigma_k^2 \right \} \\
& \ \ \ \ \ \ \ \ \ \ \ \ \ \ \ \ \ \ \ \ \ \ \ \ \ \ \ \ \ \ \ \ \ \ \ \ \ \ \ \ \ \ \ - N \overline{\mathrm{log} \left ( s^2 \right )} \sum_{i=1}^N \mathrm{log} \left \{ \left ( 1 - h_{ii} \right )^2 \sigma_i^2 + \sum_{k \neq i} h_{ik}^2 \sigma_k^2 \right \}  \Bigg{)} \\
& \ \ \ + \frac{1}{N \displaystyle \sum_{k=1}^N \left \{ \mathrm{log} \left ( s_k^2 \right ) - \overline{\mathrm{log} \left ( s^2 \right )} \right \}^2} \Bigg{(} N \displaystyle \sum_{k=1}^N \mathrm{log} \left ( s_k^2 \right ) \mathbb{E} \left \{ \mathrm{log} \left ( Z_k^2 \right ) \mid \mathbf{U}, \mathbf{V} \right \} \\
& \ \ \ \ \ \ \ \ \ \ \ \ \ \ \ \ \ \ \ \ \ \ \ \ \ \ \ \ \ \ \ \ \ \ \ \ \ \ \ \ \ \ \ - N\overline{\mathrm{log} \left ( s^2 \right )} \displaystyle \sum_{k=1}^N \mathbb{E} \left \{ \mathrm{log} \left ( Z_k^2 \right ) \mid \mathbf{U}, \mathbf{V} \right \} \Bigg{)}.
\end{align*}
\noindent Now
\begin{align*}
\frac{1}{N \displaystyle \sum_{k=1}^N \left \{ \mathrm{log} \left ( s_k^2 \right ) - \overline{\mathrm{log} \left ( s^2 \right )} \right \}^2} & \Bigg{(} \left [ \displaystyle \sum_{k=1}^N \left \{ \mathrm{log} \left ( s_k^2 \right ) \right \}^2 \right ] \displaystyle \sum_{k=1}^N \mathbb{E} \left \{ \mathrm{log} \left ( Z_k^2 \right ) \mid \mathbf{U}, \mathbf{V} \right \} \\
& \ \ \ - N \overline{\mathrm{log} \left ( s^2 \right )} \displaystyle \sum_{k=1}^N \mathrm{log} \left ( s_k^2 \right ) \mathbb{E} \left \{ \mathrm{log} \left ( Z_k^2 \right ) \mid \mathbf{U}, \mathbf{V} \right \} \Bigg{)} = - \left ( \gamma_{\mathrm{EM}} + \mathrm{log} 2 \right ),
\end{align*}
so that
\begin{align*}
\mathbb{E} \left ( \hat{\gamma}_0 \mid \mathbf{U}, \mathbf{V} \right ) &= 
\frac{1}{N \displaystyle \sum_{k=1}^N \left \{ \mathrm{log} \left ( s_k^2 \right ) - \overline{\mathrm{log} \left ( s^2 \right )} \right \}^2} \Bigg{(} \left [ \displaystyle \sum_{k=1}^N \left \{ \mathrm{log} \left ( s_k^2 \right ) \right \}^2 \right ] \displaystyle \sum_{i=1}^N \mathrm{log} \left \{ \left ( 1 - h_{ii} \right )^2 \sigma_i^2 + \displaystyle \sum_{k \neq i} h_{ik}^2 \sigma_k^2 \right \} \\
& \ \ \ \ \ \ \ \ \ \ \ \ \ \ \ \ \ \ \ \ \ \ \ \ \ \ \ \ \ \ \ \ \ \ \ \ \ \ \ \ \ \ \ - N\overline{\mathrm{log} \left ( s^2 \right )}  \displaystyle \sum_{i=1}^N \mathrm{log} \left ( s_i^2 \right ) \mathrm{log} \left \{ \left ( 1 - h_{ii} \right )^2 \sigma_i^2 + \displaystyle \sum_{k \neq i} h_{ik}^2 \sigma_k^2 \right \} \Bigg{)} \\
& \ \ \ - \left ( \gamma_{\mathrm{EM}} + \mathrm{log} 2 \right ).
\end{align*}
Also, as
\[
N \displaystyle \sum_{k=1}^N \mathrm{log} \left ( s_k^2 \right ) \mathbb{E} \left \{ \mathrm{log} \left ( Z_k^2 \right ) \mid \mathbf{U}, \mathbf{V} \right \} - N\overline{\mathrm{log} \left ( s^2 \right )} \displaystyle \sum_{k=1}^N \mathbb{E} \left \{ \mathrm{log} \left ( Z_k^2 \right ) \mid \mathbf{U}, \mathbf{V} \right \} = 0
\]
we have that
\begin{align*}
\mathbb{E} \left ( \hat{\gamma}_1 \mid \mathbf{U}, \mathbf{V} \right ) &= \frac{1}{N \displaystyle \sum_{k=1}^N \left \{ \mathrm{log} \left ( s_k^2 \right ) - \overline{\mathrm{log} \left ( s^2 \right )} \right \}^2} \Bigg{[} N \displaystyle \sum_{i=1}^N \mathrm{log} \left ( s_i^2 \right ) \mathrm{log} \left \{ \left ( 1 - h_{ii} \right )^2 \sigma_i^2 + \sum_{k \neq i} h_{ik}^2 \sigma_k^2 \right \} \\
& \ \ \ \ \ \ \ \ \ \ \ \ \ \ \ \ \ \ \ \ \ \ \ \ \ \ \ \ \ \ \ \ \ \ \ \ \ \ \ \ \ \ \ - N \overline{\mathrm{log} \left ( s^2 \right )} \sum_{i=1}^N \mathrm{log} \left \{  \left ( 1 - h_{ii} \right )^2 \sigma_i^2 + \sum_{k \neq i} h_{ik}^2 \sigma_k^2 \right \}  \Bigg{]}.
\end{align*}
\end{proof}

\begin{proof}[Proof of Proposition \ref{lem:covariances}]
We calculate the finite-sample covariance matrix of the skedastic function model coefficient estimators via
\[
\mathrm{Cov} \left \{ \left ( \mathbf{U}^{\mathsf{T}}\mathbf{U} \right )^{-1}\mathbf{U}^{\mathsf{T}} \begin{pmatrix} \mathrm{log} \left ( Z_1^2 \right ) \\ \vdots \\ \mathrm{log} \left ( Z_N^2 \right ) \end{pmatrix} \bigg| \mathbf{U}, \mathbf{V} \right \} = \left ( \mathbf{U}^{\mathsf{T}}\mathbf{U} \right )^{-1}\mathbf{U}^{\mathsf{T}} \mathrm{Cov} \left ( \begin{pmatrix} \mathrm{log} \left ( Z_1^2 \right ) \\ \vdots \\ \mathrm{log} \left ( Z_N^2 \right ) \end{pmatrix} \bigg| \mathbf{U}, \mathbf{V} \right ) \mathbf{U} \left ( \mathbf{U}^{\mathsf{T}}\mathbf{U} \right )^{-1}. 
\]
We have $\mathrm{Var} \left \{ \mathrm{log} \left ( Z_i^2 \right ) \mid \mathbf{U}, \mathbf{V} \right \} = \pi^2/2$ and $\bigg| \mathrm{Cov} \left \{ \mathrm{log} \left ( Z_i^2 \right ), \mathrm{log} \left ( Z_j^2 \right )  \mid \mathbf{U}, \mathbf{V} \right \} \bigg| < \infty$ for all $i, j \in \{1, \ldots, N \}$. Let
\[
C = \mathrm{Cov} \left ( \begin{pmatrix} \mathrm{log} \left ( Z_1^2 \right ) \\ \vdots \\ \mathrm{log} \left ( Z_N^2 \right ) \end{pmatrix} \bigg| \mathbf{U}, \mathbf{V} \right ),
\]
and let $C_{ij}$ denote the $(i,j)$ entry of $C$ (where $C_{ij} = C_{ji}$) for all $i, j \in \{ 1, \ldots, N \}$. Then $\mathbf{U}^{\mathsf{T}}C\mathbf{U}$ is
\[
\begin{pmatrix} \displaystyle \sum_{i=1}^N \sum_{j=1}^N C_{ij} & \displaystyle \sum_{i=1}^N \sum_{j=1}^N C_{ij} \mathrm{log} \left ( s_i^2 \right ) \\ \displaystyle \sum_{i=1}^N \sum_{j=1}^N C_{ij} \mathrm{log} \left ( s_i^2 \right ) & \displaystyle \sum_{i=1}^N \sum_{j=1}^N C_{ij} \mathrm{log} \left ( s_i^2 \right ) \mathrm{log} \left ( s_j^2 \right ) \end{pmatrix},
\]
which is
\[
\begin{pmatrix} \mathrm{Var} \left \{ \displaystyle \sum_{i=1}^N \mathrm{log} \left ( Z_i^2 \right ) \mid \mathbf{U}, \mathbf{V} \right \} & \mathrm{Cov} \left \{ \displaystyle \sum_{i=1}^N \mathrm{log} \left ( s_i^2 \right ) \mathrm{log} \left ( Z_i^2 \right ), \displaystyle \sum_{j=1}^N \mathrm{log} \left ( Z_j^2 \right ) \mid \mathbf{U}, \mathbf{V} \right \} \\ \mathrm{Cov} \left \{ \displaystyle \sum_{i=1}^N \mathrm{log} \left ( s_i^2 \right ) \mathrm{log} \left ( Z_i^2 \right ), \displaystyle \sum_{j=1}^N \mathrm{log} \left ( Z_j^2 \right ) \mid \mathbf{U}, \mathbf{V} \right \} & \mathrm{Var} \left \{ \displaystyle \sum_{i=1}^N \mathrm{log} \left ( s_i^2 \right ) \mathrm{log} \left ( Z_i^2 \right ) \mid \mathbf{U}, \mathbf{V} \right \} \end{pmatrix}.
\]
\noindent As such, 
\begin{align*}
&\mathrm{Cov} \left \{ \begin{pmatrix} \hat{\gamma}_0 \\ \hat{\gamma}_1 \end{pmatrix} \bigg| \mathbf{U}, \mathbf{V} \right \} = \left ( \mathbf{U}^{\mathsf{T}}\mathbf{U} \right )^{-1} \left ( \mathbf{U}^{\mathsf{T}} C \mathbf{U} \right ) \left ( \mathbf{U}^{\mathsf{T}} \mathbf{U} \right )^{-1} \\
&= \left [ N \displaystyle \sum_{k=1}^N \left \{ \mathrm{log} \left ( s_k^2 \right ) - \overline{\mathrm{log} \left ( s^2 \right )} \right \}^2 \right ]^{-2} \begin{pmatrix} \displaystyle \sum_{k=1}^N \left \{ \mathrm{log} \left ( s_k^2 \right ) \right \}^2 & -N \overline{\mathrm{log} \left ( s^2 \right )} \\ -N \overline{\mathrm{log} \left ( s^2 \right )} & N \end{pmatrix} \\
& \ \ \ \times \begin{pmatrix} \mathrm{Var} \left \{ \displaystyle \sum_{i=1}^N \mathrm{log} \left ( Z_i^2 \right ) \mid \mathbf{U}, \mathbf{V} \right \} & \mathrm{Cov} \left \{ \displaystyle \sum_{i=1}^N \mathrm{log} \left ( s_i^2 \right ) \mathrm{log} \left ( Z_i^2 \right ), \displaystyle \sum_{j=1}^N \mathrm{log} \left ( Z_j^2 \right ) \mid \mathbf{U},\mathbf{V} \right \} \\ \mathrm{Cov} \left \{ \displaystyle \sum_{i=1}^N \mathrm{log} \left ( s_i^2 \right ) \mathrm{log} \left ( Z_i^2 \right ), \displaystyle \sum_{j=1}^N \mathrm{log} \left ( Z_j^2 \right ) \mid \mathbf{U}, \mathbf{V} \right \} & \mathrm{Var} \left ( \displaystyle \sum_{i=1}^N \mathrm{log} \left \{ s_i^2 \right ) \mathrm{log} \left ( Z_i^2 \right ) \mid \mathbf{U}, \mathbf{V} \right \} \end{pmatrix} \\
& \ \ \ \times \begin{pmatrix} \displaystyle \sum_{k=1}^N \left \{ \mathrm{log} \left ( s_k^2 \right ) \right \}^2 & -N \overline{\mathrm{log} \left ( s^2 \right )} \\ -N \overline{\mathrm{log} \left ( s^2 \right )} & N \end{pmatrix}.
\end{align*}
In particular,
\[
\mathrm{Var} \left ( \hat{\gamma}_0 \mid \mathbf{U}, \mathbf{V} \right ) = \frac{\mathrm{Var} \left [ \left \{ \displaystyle \sum_{k=1}^N \mathrm{log} \left ( s_k^2 \right ) \right \} \displaystyle \sum_{i=1}^N \mathrm{log} \left ( Z_i^2 \right ) - N \overline{\mathrm{log} \left ( s^2 \right )} \displaystyle \sum_{j=1}^N \mathrm{log} \left ( s_j^2 \right ) \mathrm{log} \left ( Z_j^2 \right ) \mid \mathbf{U}, \mathbf{V} \right ]}{\left [ N \displaystyle \sum_{k=1}^N \left \{ \mathrm{log} \left ( s_k^2 \right ) - \overline{\mathrm{log} \left ( s^2 \right )} \right \}^2 \right ]^{2}},
\]
\[
\mathrm{Var} \left ( \hat{\gamma}_1 \mid \mathbf{U}, \mathbf{V} \right ) = \frac{\mathrm{Var} \left \{ N \overline{\mathrm{log} \left ( s^2 \right )} \displaystyle \sum_{i=1}^N \mathrm{log} \left ( Z_i^2 \right ) - N \displaystyle \sum_{j=1}^N \mathrm{log} \left ( s_j^2 \right ) \mathrm{log} \left ( Z_j^2 \right ) \mid \mathbf{U}, \mathbf{V} \right \}}{\left [ N \displaystyle \sum_{k=1}^N \left \{ \mathrm{log} \left ( s_k^2 \right ) - \overline{\mathrm{log} \left ( s^2 \right )} \right \}^2 \right ]^{2}} ,
\]
\noindent and
\begin{align*}
\mathrm{Cov} \left ( \hat{\gamma}_0, \hat{\gamma}_1 \mid \mathbf{U}, \mathbf{V} \right ) &= \left [ N \displaystyle \sum_{k=1}^N \left \{ \mathrm{log} \left ( s_k^2 \right ) - \overline{\mathrm{log} \left ( s^2 \right )} \right \}^2 \right ]^{-2} \\
& \ \ \ \times  \Bigg{[} -N \overline{\mathrm{log} \left ( s^2 \right )} \mathrm{Var} \left ( \sum_{i=1}^N \mathrm{log} \left ( Z_i^2 \right ) \mid \mathbf{U}, \mathbf{V} \right ) \sum_{k=1}^N \left \{ \mathrm{log} \left ( s_k^2 \right ) \right \}^2 \\
& \ \ \ \ \ \ \ \ + N^2 \left \{ \overline{\mathrm{log} \left ( s^2 \right )} \right \}^2 \mathrm{Cov} \left \{ \sum_{i=1}^N \mathrm{log} \left ( Z_i^2 \right ), \sum_{j=1}^N \mathrm{log} \left ( s_j^2 \right ) \mathrm{log} \left ( Z_j^2 \right ) \mid \mathbf{U}, \mathbf{V} \right \} \\
& \ \ \ \ \ \ \ \ + N \mathrm{Cov} \left \{ \sum_{i=1}^N \mathrm{log} \left ( Z_i^2 \right ), \sum_{j=1}^N \mathrm{log} \left ( s_j^2 \right ) \mathrm{log} \left ( Z_j^2 \right ) \mid \mathbf{U}, \mathbf{V} \right \} \sum_{k=1}^N \left \{ \mathrm{log} \left ( s_k^2 \right ) \right \}^2 \\
& \ \ \ \ \ \ \ \ - N^2 \overline{\mathrm{log} \left ( s^2 \right )} \mathrm{Var} \left \{ \sum_{i=1}^N \mathrm{log} \left ( s_i^2 \right ) \mathrm{log} \left ( Z_i^2 \right ) \mid \mathbf{U}, \mathbf{V} \right \} \Bigg{]}.
\end{align*}
\end{proof}

\begin{proof}[Proof of Theorem \ref{thm:finite_sample_beta_expectation}]
First, 
\begin{align*}
\mathbb{E} \left ( \hat{\beta} \mid \mathbf{U}, \mathbf{V} \right ) &= \mathbb{E} \left \{ \mathbb{E} \left ( \hat{\beta} \mid e_1^2, \ldots, e_N^2, \mathbf{U}, \mathbf{V} \right ) \mid \mathbf{U}, \mathbf{V} \right \} \\
&= \mathbb{E} \left \{ \left ( \mathbf{V}^{\mathsf{T}} \hat{\Omega}^{-1} \mathbf{V} \right )^{-1} \mathbf{V}^{\mathsf{T}} \hat{\Omega}^{-1} \mathbb{E} \left ( y \mid e_1^2, \ldots, e_N^2, \mathbf{U}, \mathbf{V} \right ) \mid \mathbf{U}, \mathbf{V} \right \} \\
&= \beta + \mathbb{E} \left \{ \left ( \mathbf{V}^{\mathsf{T}} \hat{\Omega}^{-1} \mathbf{V} \right )^{-1} \mathbf{V}^{\mathsf{T}} \hat{\Omega}^{-1} \mathbb{E} \left ( \epsilon \mid e_1^2, \ldots, e_N^2, \mathbf{U}, \mathbf{V} \right ) \mid \mathbf{U}, \mathbf{V} \right \},
\end{align*}
where $\epsilon = \left ( \epsilon_1, \ldots, \epsilon_N \right )^{\mathsf{T}}$. Next, we have from equation (\ref{eq:linear_model_observed_outcomes}) that 
\[
\begin{pmatrix} e_1^2 \\ \vdots \\ e_N^2 \end{pmatrix} =  \begin{pmatrix} \left \{ d_1^{\mathsf{T}} \left ( I - \mathbf{H} \right ) \epsilon \right \}^2 \\ \vdots \\ \left \{ d_N^{\mathsf{T}} \left ( I - \mathbf{H} \right ) \epsilon \right \}^2 \end{pmatrix}.
\]
\noindent Given $e_1^2, \ldots, e_N^2$, if a vector $\epsilon^*$ satisfies this set of equations then $-\epsilon^*$ also satisfies the set of equations. Accordingly, each $\left [ \epsilon_i \mid e_1^2, \ldots, e_N^2, \mathbf{U}, \mathbf{V} \right ]$ takes only finitely many possible values that are symmetric around zero. Furthermore, as $\left [ \epsilon_i \mid \mathbf{U}, \mathbf{V} \right ]$ is a symmetric distribution centered at zero, the probability mass function of $\left [ \epsilon_i \mid e_1^2, \ldots, e_N^2, \mathbf{U}, \mathbf{V} \right ]$ is also symmetric around zero. Thus $\mathbb{E} \left ( \epsilon_i \mid e_1^2, \ldots, e_N^2, \mathbf{U}, \mathbf{V} \right ) = 0$ so that $\mathbb{E} \left ( \hat{\beta} \mid \mathbf{U}, \mathbf{V} \right ) = \beta$.
\end{proof}

\begin{proof}[Proof of Proposition \ref{lem:consistency_expectation}]
From the Weak Law of Large Numbers and the assumptions, we have 
\[
N^{-1} \sum_{k=1}^N \left \{ \mathrm{log} \left ( s_k^2 \right ) - \overline{\mathrm{log} \left ( s^2 \right )} \right \} \rightarrowp \mathrm{Var} \left \{ \mathrm{log} \left ( S^2 \right ) \right \}
\]
and
\begin{align*}
&\frac{1}{N} \sum_{i=1}^N \left \{ \mathrm{log} \left ( s_i^2 \right ) - \overline{\mathrm{log} \left ( s^2 \right )} \right \} \mathrm{log} \left \{ \left ( 1 - h_{ii} \right )^2 \sigma_i^2 + \sum_{k \neq i} h_{ik}^2 \sigma_k^2 \right \} \\
&= \frac{1}{N} \sum_{i=1}^N \left \{ \mathrm{log} \left ( s_i^2 \right ) - \overline{\mathrm{log} \left ( s^2 \right )} \right \} \Bigg{[} \mathrm{log} \left \{  \left ( 1 - h_{ii} \right )^2 \sigma_i^2 + \sum_{k \neq i} h_{ik}^2 \sigma_k^2 \right \} \\
& \ \ \ \ \ \ \ \ \ \ \ \ \ \ - \frac{1}{N} \sum_{j=1}^N \mathrm{log} \left \{ \left ( 1 - h_{jj} \right )^2 \sigma_j^2 + \sum_{k \neq j} h_{jk}^2 \sigma_k^2 \right \} \Bigg{]} \\
&\rightarrowp \mathrm{Cov} \left ( \mathrm{log} \left ( S^2 \right ), \mathrm{log} \left ( H_{\sigma} \right ) \right ).
\end{align*}
\noindent By virtue of the Continuous Mapping Theorem, $\mathbb{E} \left ( \hat{\gamma}_1 \mid \mathbf{U}, \mathbf{V} \right ) \rightarrowp$ $\mathrm{Cov} \left \{ \mathrm{log} \left ( S^2 \right ), \mathrm{log} \left ( H_{\sigma} \right ) \right \}$ $/$ $\mathrm{Var} \left \{ \mathrm{log} \left ( S^2 \right ) \right \}$, and so from the portmanteau lemma $\mathbb{E} \left ( \hat{\gamma}_1 \right ) = \mathbb{E} \left \{ \mathbb{E} \left ( \hat{\gamma}_1 \mid \mathbf{U}, \mathbf{V} \right ) \right \} \rightarrow$ $\mathrm{Cov} \left \{ \mathrm{log} \left ( S^2 \right ), \mathrm{log} \left ( H_{\sigma} \right ) \right \}/\mathrm{Var} \left \{ \mathrm{log} \left ( S^2 \right ) \right \}$ as $N \rightarrow \infty$.

A similar argument applies to $\mathbb{E} \left ( \hat{\gamma}_0 \right )$. In particular, we have from the Weak Law of Large Numbers and the Continuous Mapping Theorem that
\[
\frac{1}{N^2} \left [ \sum_{k=1}^N \left \{ \mathrm{log} \left ( s_k^2 \right ) \right \}^2 \right ] \sum_{i=1}^N \mathrm{log} \left \{ \left ( 1 - h_{ii} \right )^2 \sigma_i^2 + \sum_{k \neq i} h_{ik}^2 \sigma_k^2 \right \} \rightarrowp \mathbb{E} \left [ \left \{ \mathrm{log} \left ( S^2 \right ) \right \}^2 \right ] \mathbb{E} \left \{ \mathrm{log} \left ( H_{\sigma} \right ) \right \}
\]
and
\[
\frac{1}{N} \left \{ \overline{\mathrm{log} \left ( s^2 \right )} \right \} \sum_{i=1}^N \left [ \mathrm{log} \left ( s_i^2 \right ) \mathrm{log} \left \{ \left ( 1 - h_{ii} \right )^2 \sigma_i^2 + \sum_{k \neq i} h_{ik}^2 \sigma_k^2 \right \} \right ] \rightarrowp \mathbb{E} \left \{ \mathrm{log} \left ( S^2 \right ) \right \} \mathbb{E} \left \{ \mathrm{log} \left ( S^2 \right ) \mathrm{log} \left ( H_{\sigma} \right ) \right \},
\]
so that 
\[
\mathbb{E} \left ( \hat{\gamma}_0 \mid \mathbf{U}, \mathbf{V} \right ) \rightarrowp \frac{\mathbb{E} \left [ \left \{ \mathrm{log} \left ( S^2 \right ) \right \}^2 \right ] \mathbb{E} \left \{ \mathrm{log} \left ( H_{\sigma} \right ) \right \} - \mathbb{E} \left \{ \mathrm{log} \left ( S^2 \right ) \right \} \mathbb{E} \left \{ \mathrm{log} \left ( S^2 \right ) \mathrm{log} \left ( H_{\sigma} \right ) \right \}}{\mathrm{Var} \left \{ \mathrm{log} \left ( S^2 \right ) \right \}} - \left ( \gamma_{\mathrm{EM}} + \mathrm{log} 2 \right ).
\]
\noindent Again, the portmanteau lemma implies that as $N \rightarrow \infty$
\begin{align*}
\mathbb{E} \left ( \hat{\gamma}_0 \right ) &= \mathbb{E} \left \{ \mathbb{E} \left ( \hat{\gamma}_0 \mid \mathbf{U}, \mathbf{V} \right ) \right \} \\
& \ \ \ \rightarrow \frac{\mathbb{E} \left [ \left \{ \mathrm{log} \left ( S^2 \right ) \right \}^2 \right ] \mathbb{E} \left \{ \mathrm{log} \left ( H_{\sigma} \right ) \right \} - \mathbb{E} \left \{ \mathrm{log} \left ( S^2 \right ) \right \} \mathbb{E} \left \{ \mathrm{log} \left ( S^2 \right ) \mathrm{log} \left ( H_{\sigma} \right ) \right \}}{\mathrm{Var} \left \{ \mathrm{log} \left ( S^2 \right ) \right \}} - \left ( \gamma_{\mathrm{EM}} + \mathrm{log} 2 \right ).
\end{align*}
\end{proof}

\bibliographystyle{apalike}
\bibliography{references}

\end{document}